\documentclass[a4paper,11pt]{article}

\usepackage[british,UKenglish]{babel}
\usepackage[utf8x]{inputenc}
\usepackage[T1]{fontenc}

\usepackage{tikz-cd}

\usepackage[a4paper,top=3cm,bottom=2cm,left=3cm,right=3cm,marginparwidth=1.75cm]{geometry}

\usepackage{amsmath,amsfonts,amssymb}
\usepackage{graphicx}
\usepackage{amsthm}
\usepackage{xcolor}
\usepackage{cite}

\newtheorem{theorem}{\textsc{Theorem}}[section]
\newtheorem{lemma}{\textsc{Lemma}}[section]
\newtheorem{definition}{\textsc{Definition}}[section]

\newtheorem{corollary}{\textsc{Corollary}}[section]
\newtheorem{proposition}{\textsc{Proposition}}[section]
\newtheorem{remark}{\textsc{Remark}}[section]

\newtheorem{condition}{\textsc{Condition}}[section]

\usepackage{authblk}

\def\dom{\text{dom}}
\def\sign{\text{sgn}}
\def\deg{\text{deg}}
\def\esssup{\text{esssup}}
\def\Re{\text{Re}}
\def\range{\text{range}}

\newcommand{\p}{\partial}
\newcommand{\bb}{\begin{equation}}
\newcommand{\ee}{\end{equation}}
\newcommand{\ba}{\begin{array}}
\newcommand{\ea}{\end{array}}
\newcommand{\R}{\mathbb{R}}
\newcommand{\f}{\frac}
\usepackage[all]{xy}
\newcommand{\ds}{\displaystyle}

\newcommand{\al}{\alpha}
\newcommand{\be}{\beta}

\numberwithin{equation}{subsection}

\makeatletter
\def\@maketitle{%
  \newpage
  \null
  \vskip 2em%
  \begin{center}%
  \let \footnote \thanks
    {\huge \@title \par}%
    \vskip 1.5em%
    {\large
      \lineskip .5em%
      \begin{tabular}[t]{c}%
        \@author
      \end{tabular}\par}%
    \vskip 1em%
  \end{center}%
  \par
  \vskip 1.5em
  \vspace{1cm}}
\makeatother

\begin{document}

\title{Well-posedness, travelling waves and geometrical aspects of generalizations of the Camassa-Holm equation\vspace{1cm}}
\author[1]{Priscila Leal da Silva\thanks{pri.leal.silva@gmail.com}}
\author[2,3]{Igor Leite Freire\thanks{igor.freire@ufabc.edu.br and igor.leite.freire@gmail.com}}
\affil[1]{Departamento de Matem\'atica, Universidade Federal de S\~ao Carlos, Brasil}
\affil[2]{Mathematical Institute, Silesian University in Opava,
Na Rybn\'\i{}\v{c}ku, 1, 74601, Opava, Czech Republic}
\affil[3]{Centro de Matem\'atica, Computa\c{c}\~ao e Cogni\c c\~ao, Universidade Federal do ABC, Avenida dos Estados, $5001$, Bairro Bangu,
$09.210-580$, Santo Andr\'e, SP - Brasil}
\maketitle

\renewcommand{\abstractname}{\vspace{-\baselineskip}}

\begin{abstract}
\centering\begin{minipage}{\dimexpr\paperwidth-9cm}
\textbf{Abstract:} In this paper we consider a five-parameter equation including the Camassa-Holm and the Dullin-Gottwald-Holm equations, among others. We prove the existence and uniqueness of solutions of the Cauchy problem using Kato's approach. Conservation laws of the equation, up to second order, are also investigated. From these conservation laws we establish some properties for the solutions of the equation and we also find a quadrature for it. The quadrature obtained is of capital importance in a classification of bounded travelling wave solutions. We also find some explicit solutions, given in terms of elliptic integrals. Finally, we classify the members of the equation describing pseudo-spherical surfaces.

\vspace{0.2cm}
\textbf{2010 AMS Mathematics Classification numbers}: 35A01, 35L65, 37K05.


\textbf{Dedicatory:} This paper is dedicated to Professor Nail Ibragimov.

\end{minipage}
\end{abstract}
\bigskip
\newpage
\tableofcontents
\newpage

\section{Introduction}

The Camassa-Holm (CH) equation 
$$m_t+um_x+2u_xm=0,\quad m=u-u_{xx},$$ was named after the pioneering work of Camassa and Holm \cite{chprl}. Despite the fact that the equation itself was discovered earlier from the investigation of equations having Hamiltonian representations \cite{fofu}, it was the work of Camassa and Holm that derived the equation with a physical background and showed its relevance in the physics of fluids. Since then, this sort of equation has been subject of intense research, which is easy to understand due to its physical relevance \cite{chprl,dgh,john}, and also the rich mathematical structures behind the equation \cite{const1998-1, const1998-2, const1998-3,const2000-1,const2000-2,const2002, escher,HKM,len, strachan, reyes2002,blanco}.

From the myriad of properties of the CH equation, we point out its solutions \cite{const1998-1, const1998-2, const1998-3,const2000-1,const2000-2,const2002, HKM, blanco} and, in particular, the fashion and attractive peakon solutions \cite{const2000-2,len}. However, its algebraic properties \cite{strachan,novikov} and geometrical aspects \cite{escher,reyes2002} are equally rich and interesting. 

Since its derivation, other equations sharing common properties with the CH equation have been discovered and investigated from both mathematical or physical motivations \cite{anco,pri-book,pri-aims,deg,dgh,hone,novikov,chines-arxiv,chines-jde,gui-jnl}, which are also subject of interest for a wide spectra of researchers and from different frames, see \cite{anco,boz,pri-book,pri-aims,raspa,kelly,hak,HKM,HM,liu2011,novikov,mustafa,reyes2011,blanco,strachan} and references therein. Moreover, modifications, extensions and generalizations of these equations have also been intensely studied, see \cite{anco,pri-aims,hak,novikov,mustafa,liu2011,strachan} and references thereof.  

More recently, a Camassa-Holm type equation incorporating Coriolis effects\footnote{We are following the notation in \cite{chines-jde}. In the other references the constants appearing in the equation \eqref{1.0.1} are slightly different, but the difference is a scaling in space-time variables} was proposed \cite{chines-arxiv,chines-jde,chines-adv,gui-jnl}:
\bb\label{1.0.1}
m_t+um_x+2u_xm+cu_x-\f{\beta_0}{\beta}u_{xxx}+\f{\omega_1}{\al^2}u^2u_x+\f{\omega_2}{\alpha^3}u^3u_x=0,
\ee
where $m=u-u_{xx}$,
\bb\label{1.0.2}
\ba{l}
\ds{c=\sqrt{1+\Omega^2}-\Omega,\quad\al=\f{c^2}{1+c^2},\quad\beta_0=\f{c(c^4+6c^2-1)}{6(c^2+1)^2},\quad \beta=\f{3c^4+8c^2-1}{6(c^2+1)^2}},\\
\\
\ds{\omega_1=-\f{3c(c^2-1)(c^2-2)}{2(1+c^2)^3},\quad\omega_2=\f{(c^2-1)^2(c^2-2)(8c^2-1)}{2(1+c^2)^5}}
\ea
\ee
and $\Omega$ is a parameter related to the Coriolis effect. We recall that the  Coriolis effect is typically a manifestation of rotation when Newton's laws are applied to model physical phenomena on Earth's surface. 

Throughout this paper, both $\p_z(\cdot)$ or $(\cdot)_z$ mean partial derivative with respect to a generic variable $z$. In (\ref{1.0.1}), $u=u(x,t)$ is the dependent variable, which may physically describe the elevation of the water surface, while $x$ and $t$ are variables related to space and time, respectively. In the remaining of our work, we shall maintain these variables for the equations we want to investigate.

In view of the constants in \eqref{1.0.1}, we observe that the $\be_0$, $\omega_1$ and $\omega_2$ in \eqref{1.0.2} cannot vanish simultaneously and $c>0$. Therefore, although from a physical framework it is relevant to consider \eqref{1.0.1} with the parameters \eqref{1.0.2}, from a mathematical perspective these constraints impose restrictions on \eqref{1.0.1}. 

In \cite{gui-jnl}, the authors classified the types of travelling wave solutions that \eqref{1.0.1} can admit. The classification carried out in that reference is influenced, as should be expected, by the constraints \eqref{1.0.2}. However, the classification in \cite{gui-jnl} was made under a very restrictive condition, which led to a {\it partial} classification of the travelling wave solutions admitted by the equation. For example, periodic waves are not obtained in their results, and this is due to a very specific choice they made in \cite{gui-jnl} to analyse solutions.

One of the interesting properties of the CH equation is the fact that it describes pseudo-spherical surfaces, as proved by Reyes in \cite{reyes2002}, see also \cite{reyes2011}. Then, a natural question about \eqref{1.0.1}, since it generalizes the CH equation, is if it shares this property with the CH equation. This point was not investigated in \cite{chines-arxiv,chines-jde,chines-adv,gui-jnl}.

Another thought-provoking point is the local well-posedness of the equation \eqref{1.0.1} once it is well-known that the CH equation is well-posed under mild conditions on the initial data ({\it e.g.}, see \cite{escher,blanco,mustafa}). This result can be proved using Kato's theory \cite{kato}, as it was done in the mentioned references. However, in \cite{gui-jnl} the local well-posedness of \eqref{1.0.1} is claimed, but not proved, whereas in \cite{chines-arxiv,chines-jde,chines-adv} it is considered in a different perspective.

A third appealing aspect of the CH equation is the existence of peakons, which are weak peaked solutions \cite{chprl,len}. In \cite{chprl} it was shown the solitonic behaviour of certain non-periodic peakon solutions (see also \cite{anco,deg}), whereas in \cite{len} periodic peakon solutions to the CH equation were proven to exist. Regarding equation \eqref{1.0.1}, differently from the CH equation and other CH type equations deduced over time, it does not admit peakons with decay shaping like $e^{-|x|}$ , see \cite{gui-jnl}, whilst the existence of periodic peakons of \eqref{1.0.1} was not considered previously.

This paper is strongly motivated by references \cite{chines-jde,chines-arxiv}, and here we consider the equation
\bb\label{1.0.3}
\ba{l}
m_t+u\,m_x+2u_xm=\al u_x+\be u^2u_x+\gamma u^3u_x+\Gamma  u_{xxx},\\
\\
m=u-\epsilon^2u_{xx},\quad u=u(x,t),\quad (\epsilon,\Gamma)\neq(0,0)
\ea
\ee
from the point of view of local well-posedness, existence of bounded travelling wave solutions and geometrical integrability. In \eqref{1.0.3}, the real parameters $\al,\,\be,\,\gamma,\,\Gamma$ and $\epsilon$ are independently considered, that is, no dependences among them are being imposed.
For the calculations that will come, the following form of \eqref{1.0.3} is convenient:
\bb\label{1.0.4}
u_t-\epsilon^2 u_{txx}=\epsilon^2uu_{xxx}+2\epsilon^2 u_{x}u_{xx}+(\al-3u+\be u^2+\gamma u^3)u_x+\Gamma u_{xxx}.
\ee

We observe that \eqref{1.0.4} includes the CH equation ($\be=\gamma=\Gamma=0$), the Dullin-Gottwald-Holm (DGH) equation ($\be=\gamma=0$), the KdV equation ($\epsilon=\be=\gamma=0$), Gardner equation ($\epsilon=\gamma=0$) and, more generally, it is a particular case of the generalized KdV equation \cite{katoma} when $\epsilon=0$. Clearly, the restrictions \eqref{1.0.2} imply that \eqref{1.0.1} is a very particular case of \eqref{1.0.3}.

Although we avoid any restriction on the parameters in \eqref{1.0.4}, we want to impose a weak but very technical restriction: $(\epsilon,\Gamma)\neq(0,0)$. This is taken because if both $\epsilon$ and $\Gamma$ vanish, then (\ref{1.0.3}) and (\ref{1.0.4}) are reduced to a transport equation. We observe that $\epsilon$ may be $0$, which makes \eqref{1.0.3} an evolution equation.

Our first result deals with the existence and uniqueness of solutions to \eqref{1.0.4}. To achieve this goal, we make use of the machinery developed by Kato \cite{kato}, which was a tool for proving similar results to KdV type equations \cite{katoma,kato, escher}, the CH equation \cite{blanco,escher} and generalizations of the latter \cite{mustafa,liu2011}. 

The fact that $\epsilon$ is arbitrary brings some nuance to the problem and, in essence, our result of Hadamard well-posedness will depend on whether $\epsilon\neq0$ or not. Whereas we should be careful regarding this parameter, we have the following existence result for \eqref{1.0.4}:
\begin{theorem}\label{teo1.1}
Given $(\epsilon,\Gamma)\neq(0,0)$, there exist Hilbert spaces $H_1=H_1(\epsilon,\Gamma)$, $H_2=H_2(\epsilon,\Gamma)$ and $T>0$ such that, if the initial datum $u_0\in H_1$, then the problem
\bb\label{1.0.5}
\left\{\ba{l}
u_t-\epsilon^2 u_{txx}=\epsilon^2uu_{xxx}+2\epsilon^2 u_{x}u_{xx}+(\al-3u+\be u^2+\gamma u^3)u_x+\Gamma u_{xxx},\\
\\
u(x,0)=u_0(x)
\ea\right.
\ee
has a unique solution $u\in C^{0}(H_1,[0,T))\cap C^1(H_2,[0,T))$. Moreover, $T$ depends only on $\Vert u_0\Vert_{H_1}$.
\end{theorem}

As we will show, if $\epsilon=0$, then $H_1=H^s(\R)$ and $H_2=H^{s-3}(\R)$, with $s>3/2$, where $H^s(\R)$ denotes the Sobolev space of order $s$, see Subsection \ref{subsec2.1}. This is actually a particular case of the results of Kato \cite{katoma}, see Lemma \ref{lema2.6} in  Subsection \ref{subsec2.2}. What remains to be proved is the case $\epsilon\neq0$, which is more delicate. We will show that for this case we can take $H_1=H^s(\R)$ and $H_2=H^{s-1}(\R)$, $s>3/2$. 

The following result is of crucial importance for the proof of Theorem \ref{teo1.1}.
\begin{theorem}\label{teo1.2}
Let $m\geq2$ be a natural number, $s\in(3/2,m)$, and $h,\,g\in C^{m+3}(\R)$, with $h(0)=0$. If $u_0\in H^s(\R)$, there exists a maximal time $T=T(u_0)>0$ and a unique solution $u$ to the Cauchy problem
\bb\label{1.0.6}
\left\{\ba{l}
\ds{u_t-u_{txx}+\p_x h(u)=\p_x\left(\f{g'(u)}{2}u_x^2+g(u)u_{xx}\right)},\\
\\
u(x,0)=u_0(x),
\ea\right.
\ee
such that $u=u(\cdot,u_0)\in C^{0}(H^{s}(\R);[0,T))\cap C^{1}(H^{s-1}(\R),[0,T))$. Moreover, the solution depends continuously on the initial data, in the sense that the mapping $u_0\mapsto u(\cdot,u_0):H^{s}(\R)\rightarrow C^{0}(H^{s}(\R);[0,T))\cap C^{1}(H^{s-1}(\R),[0,T))$ is continuous.
\end{theorem}


A strong consequence of this result is given by
\begin{corollary}\label{cor1.1}
Assume that $m$, $s$, $h$, $g$ and $u_0$ satisfy the conditions in Theorem $\ref{teo1.2}$ and assume that $\Gamma$ is a constant. Then the initial value problem
\bb\label{1.0.7}
\left\{\ba{l}
\ds{u_t-u_{txx}+\p_x h(u)=\p_x\left(\f{g'(u)}{2}u_x^2+g(u)u_{xx}\right)+\Gamma u_{xxx}},\\
\\
u(x,0)=u_0(x),
\ea\right.
\ee
has a unique solution $u=u(\cdot,u_0)\in C^{0}(H^{s}(\R);[0,T))\cap C^{1}(H^{s-1}(\R),[0,T))$ depending continuously on the initial data.
\end{corollary}

Although similar results of Theorem \ref{teo1.2} and Corollary \ref{cor1.1} are known ({\it e.g.}, \cite{mustafa,hak,liu2011}), they impose the restriction $g(0)=0$, as in \cite{liu2011}, or take $g(u)=u$, as in \cite{mustafa}, a restriction that is eliminated in our demonstration. Therefore, Theorem \ref{teo1.2} improves and generalizes these previous results, while, at the same time, not only proves Theorem \ref{teo1.1} (which is a very particular case of Corollary \ref{cor1.1}), but also other results regarding CH type equations, see \cite{kelly, escher,blanco}. The proofs of theorems \ref{teo1.1}, \ref{teo1.2} and Corollary \ref{cor1.1} are done in Section \ref{sec2}.

In Section \ref{sec3} we consider conservation laws of \eqref{1.0.4} up to second order. The restriction to second order means that we construct vector fields whose components are functions of $t$, $x$, $u$ and derivatives of $u$ up to second order and whose divergence vanishes on the solutions of the equation. From these conservation laws, in Section \ref{sec4}, we explore some properties of the solutions of \eqref{1.0.4}. 

The conservation laws enable us to find a quadrature to \eqref{1.0.4}, which represents a cornerstone to proceed with a deep investigation of bounded travelling waves of equation \eqref{1.0.4} following the ideas developed in \cite{len}. Our classification splits in several cases, depending on the values of the parameters and the zeros of the polynomial
\bb\label{1.0.8}
P(\phi)=B+2A\phi+(c+\al)\phi^2-\phi^3+\be\phi^4+\gamma\phi^5.
\ee

It is worth emphasizing that in \cite{gui-jnl} the authors carried out a classification of bounded traveling wave solutions for the equation \eqref{1.0.1} with the restrictions given by \eqref{1.0.2}. However, during the integration process to obtain the quadrature form of the equation (namely equation $(5.4)$ in the aforementioned paper), the constants of integration are neglected in order to obtain a polynomial equivalent to
\bb\label{1.0.9}
    p(\phi)=\phi^2(a_3\phi^3+a_2\phi^2+a_1\phi+a_0),
\ee
where $a_3,a_2,a_1$ and $a_0$ are given coefficients, eventually depending on the constants of \eqref{1.0.1}. Therefore, by neglecting the constants of integration, the authors cannot guarantee the existence of periodic solutions (see Section \ref{sec5} for more details).

In our work, we shall classify the bounded travelling waves of equation \eqref{1.0.4} using \eqref{1.0.8}. We observe that the set of zeros of \eqref{1.0.8} englobes those of \eqref{1.0.9}. and as a consequence, we have:
\begin{itemize}
\item 139 cases analysed, see theorems \ref{teo5.1} -- \ref{teo5.10};
\item classification of bounded travelling waves for the case $\epsilon=0$, see Subsection 5.3;
\item classification of the bounded travelling waves for the case $\epsilon\neq0$, see Subsection 5.4. Here we not only recover the results in \cite{gui-jnl}, but as already mentioned, we classify the periodic waves, which have not been previously considered.
\end{itemize}

In Section \ref{sec6} we find some explicit travelling wave solutions for \eqref{1.0.4}. We show that this equation can only admit peakon solutions shaping like $e^{-|x|}$ \cite{chprl,anco,raspa} when equation \eqref{1.0.4} is reduced to the Dullin-Gottwald-Holm equation \cite{dgh}. We also find some solutions expressed in terms of elliptic integrals.

Geometrical aspects of equation \eqref{1.0.4} are also studied in our work. More precisely, we investigate members of this class that describe pseudo-spherical surfaces (PSS) \cite{chern, keti1992,reyes2002,reyes2011,keti2015}. Such equations have a beautiful geometric structure, since the domain of their solutions can be endowed with a Riemannian metric of constant Gauss curvature ${\mathcal K}=-1$. It is known \cite{chern} that several integrable equations (see \cite{ablo1,ablo2,olverbook,sok} for a better discussion about this subject) have this property, although not all equations describing PSS are integrable, as some examples can be found in \cite{keti2015}. To pursue this goal, we make use of some recent contributions due to Silva and Tenenblat \cite{keti2015}, where they investigated equations describing PSS of the form
\bb\label{1.0.10}
u_t-u_{txx}=\lambda uu_{xxx}+G(u,u_x,u_{xx})
\ee
with associated 1-forms 
\bb\label{1.0.11}
\omega_{1}=f_{11}dx+f_{12}dt,\quad \omega_{2}=f_{21}dx+f_{22}dt,\quad \omega_{1}=f_{31}dx+f_{32}dt,
\ee
where the coefficient functions $f_{ij}$, $i=1, 2, 3$ and $j=1,2$, depend on $x,t, u$ and derivatives of $u$. In Section \ref{sec7} we will introduce and explain all information needed regarding these concepts.

We observe that if $\epsilon \neq 0$ in \eqref{1.0.4}, we can transform it into \eqref{1.0.10} by scaling time and taking the shift $u\mapsto u- \Gamma$. In the case $\epsilon=0$, however, our main ingredient
to classify the members of the class \eqref{1.0.4} (with $\epsilon=0$) describing PSS is another work, due to Rabelo and Tenenblat \cite{keti1992}. There it was investigated whether the class of evolution equations $u_t=u_{xxx}+G(u,u_x,u_{xx})$ describes PSS. In this case, we can eliminate the constant $\alpha$ by making a suitable shift in $u$, in a similar form as described above. Our main contribution regarding PSS described by members of \eqref{1.0.4}, which will be proved in Section \ref{sec7}, is the following
\begin{theorem}\label{teo1.3}
Equation $\eqref{1.0.4}$ describes PSS in the following cases:
\begin{enumerate}
\item if $\epsilon=0$, after eliminating $\alpha$, and $\omega_2=\eta dx+f_{22}dt$,

\item if $\epsilon\neq0$, after eliminating $\Gamma$, if and only if $\beta=\gamma=0$. In this case, the associated one-forms $\eqref{1.0.11}$ are given by
$$
\ba{l}
\ds{\omega_1=\left(u-u_{xx}+b\right)dx-\left[u(u-u_{xx}+b+1)+b\mp\eta u_x \right]}dt,\\
\\
\ds{\omega_2=\eta\,dx-[\eta(1+u)\mp u_x]dt},\\
\\
 \ds{\omega_3=\pm\left(u-u_{xx}+b+1\right)dx+\left[\eta u_x\pm uu_{xx}\mp(u+1)(u+b+1)\right]dt,}
\ea
$$
where $b=-1+(\eta^2-\alpha)/2$.
\end{enumerate}
\end{theorem}

If $\epsilon=0$, equation \eqref{1.0.4} becomes a member of a family of evolution equations considered in \cite{keti1992}. Then, the first part in Theorem \ref{teo1.3} is a consequence of the results proved in \cite{keti1992}. The particular case $\al=\be=\gamma=\Gamma=0$ corresponds to the CH equation and such equation is known to describe PSS, see \cite{reyes2002}.

\subsection{Contributions of the paper and its outline}

Our goal in this paper is the investigation of equation \eqref{1.0.4}, which is a mathematical extension of the model \eqref{1.0.1} recently proposed and studied in \cite{chines-jde,chines-arxiv,chines-adv}. The equation we pay attention to, however, is mathematically richer than \eqref{1.0.1} once it is not under the restrictions imposed by the constraints \eqref{1.0.2}. Our contributions in this paper can be summarized as follows:
\begin{itemize}
\item we prove a local well-posedness result to \eqref{1.0.4}. To achieve this purpose, we generalize a theorem presented in \cite{liu2011}. As a consequence of this generalization, we have the local well-posedness of \eqref{1.0.4} guaranteed, and the same result to \eqref{1.0.1} as an immediate implication. In essence, the result originally proved in \cite{liu2011} was similar to the one presented in Corollary \ref{cor1.1}, but with the strong condition $g(0)=0$. In our case, we removed this restriction, which makes our result applicable to a larger class of equations of the type \eqref{1.0.6}. Note that \eqref{1.0.7} is a particular case of the former equation under the shift $g\mapsto g+\Gamma$. We observe that in \cite{gui-jnl} the local well-posedness of \eqref{1.0.1} is claimed, but its demonstration is omitted. This part is the subject of Section \ref{sec2}.

\item we establish conservation laws for equation \eqref{1.0.4}. We observe that the task of finding conservation laws is not simple from a practical point of view, see \cite{olverbook,popPLA}. Then, we impose the following restrictions on the conserved currents we look for: they should depend up to second order derivatives, with particular emphasis to case $\epsilon\neq0$. These results will be obtained in Section \ref{sec3}. Case $\epsilon=0$ is barely considered because the conservation laws for the family obtained with this condition has been widely investigated along the last five decades. To cite a few, several conservation laws for this case can be found in \cite{ablo1,ablo2,igor2012,igor2014,miura,olverbook,popPLA,rita} and references therein. 

Moreover, if we do not impose the restriction $\epsilon=0$ we are able to find useful conservation laws to obtain qualitative information about the behaviour of solutions of \eqref{1.0.4}, such as preservation of the sign of the initial condition and a quadrature that holds for any value of $\epsilon$. These results are proven in Section \ref{sec4}.

\item The quadrature is crucial in the classification of bounded travelling wave solutions of \eqref{1.0.4}. For the case $\epsilon=0$ we have classified 25 cases of travelling wave solutions, whereas for $\epsilon\neq0$ we have classified 114 cases of wave solutions. We observe that in \cite{gui-jnl} a classification of travelling waves were carried out as well, but the quadrature the authors considered is a very restrictive case of ours. Then, our classification, presented in Section \ref{sec6}, is more complete than the one carried out in \cite{gui-jnl}. 

\item In Section \ref{sec7} we provide a complete description of members of \eqref{1.0.4} describing PSS. As an application of our result, we find the members of \eqref{1.0.1} that can describe this sort of surfaces. This is done by imposing \eqref{1.0.1} to satisfy certain conditions, which implies on restrictions on the parameters \eqref{1.0.2}. Consequently, we find the values of the physical variable $\Omega$ that allow \eqref{1.0.2} to be compatible with the fact that \eqref{1.0.1} describes PSS.

\item We discuss our results in Section \ref{sec8}, whereas in Section \ref{sec9} we present our conclusions.
\end{itemize}

\section{Well-posedness}\label{sec2}

In this section we prove Theorem \ref{teo1.1}. Our main ingredient is Kato's approach \cite{kato,katoma}, which we shall summarize in the subsection \ref{subsec2.2}. Firstly we present a short review on function space and fix the terminology. Then, in Subsection \ref{subsec2.3} we prove some technical results that will be used in subsections \ref{subsec2.4} and \ref{subsec2.5} to prove theorems \ref{teo1.2} and \ref{teo1.1}, respectively.

\subsection{Preliminaries}\label{subsec2.1}

It will be convenient to recall some terminology on function spaces and also fix some notation and terminology. For further details on these topics, see \cite{brezis,hunter,schwartz}.

The Hilbert space of all square integrable equations on the real line $\R$ is denoted by $L^2(\R)$ and is endowed with the norm
$$\|f\|_{L^{2}}=\sqrt{\int_\R |f|^2dx}.$$

More generally, given $p\in[1,\infty)$, by $L^p(\R)$ we denote the space of functions $f:\R\rightarrow\R$ such that $$\int_\R|f|^pdx<\infty.$$ It has the structure of a Banach space when endowed with the norm
$$\|f\|_{L^p}:=\sqrt[p]{\int_\R|f|^p dx}.$$
For $p=\infty$, we have the Banach space $(L^\infty(\R),\|\cdot\|_\infty)$, where
$$\|f\|_\infty:=\esssup{|f(x)|}.$$



Let $C^p_0(\R)$, $0\leq p\leq \infty$, be the set of $C^p$ functions $f:\R\rightarrow\R$ with compact support. The set of infinitely many smooth functions decaying as faster as any power to $0$ at infinity, with the same property holding to any of their derivatives, is denoted by ${\cal S}(\R)$ and is referred as the Schwartz space. We observe that $\overline{C^\infty_0(\R)}={\cal S}(\R)$ and an element of ${\cal S}(\R)$ is called {\it test function}. 

The dual topological space of ${\cal S}(\R)$ is denoted by ${\cal S}'(\R)$ and its members are called {\it tempered distributions}. Given a tempered distribution $\phi$, its Fourier transform ${\cal F}(\phi)$ is denoted by $\hat{\phi}$. Explicitly, we have
$$
\hat{\phi}(\xi)=\f{1}{\sqrt{2\pi}}\int_{-\infty}^{+\infty}\phi(x)e^{-ix\xi}dx.
$$
Moreover, its inverse is given by
$$
\phi(x)=\f{1}{\sqrt{2\pi}}\int_{-\infty}^{+\infty}\hat{\phi}(\xi)e^{ix\xi}d\xi.
$$
Very often in this work, given a function $u=u(x,t)$, we shall consider the function $x\mapsto u(x,t)$ and consider its Fourier transform, to each fixed value of $t$.

Given $s\in\R$, the space $H^{s}(\R)$ of the tempered distributions $u\in{\cal S}'(\R)$ such that $(1+|\xi|^2)^{s/2}\hat{u}(\xi)\in L^2(\R)$ is a Hilbert space when endowed with the inner product
$$\langle u,v\rangle _{H^{s}}:=\int_{\R}(1+|\xi|^2)^{s}\hat{u}(\xi)\overline{\hat{v}(\xi)}d\xi.$$

More generally, throughout this paper, if $X$ is a Banach space, its norm will be denoted by $\|\cdot\|_X$, whereas if $H$ is a Hilbert space, its inner product will be referred as $\langle\cdot,\cdot\rangle_H$.

Consider the family $\{H^s(\R),\,s\in\R\}$. We recall the following facts (see \cite{taylor}, chapter 4, or \cite{escher}):
\begin{itemize}
    \item[{\bf F1:}] We have the sequel of continuous and dense embeddings for $s\geq t$: ${\cal S}(\R)\subseteq H^{s}(\R)\subseteq H^t(\R)\subseteq{\cal S}'(\R)$;
   
    \item[{\bf F2:}] the dual of $H^s(\R)$ is $H^{-s}(\R)$, that is, $(H^{s}(\R))'=H^{-s}(\R)$;
    \item[{\bf F3:}] $\p_x:u\mapsto \p_x u:=u_x$ is a linear and continuous operator between $H^{s}(\R)$ and $H^{s-1}(\R)$;
    \item[{\bf F4:}] For each $s$, let $\Lambda^s u:={\cal F}^{-1}((1+|\xi|^2)^\f{s}{2}\hat{u})$, where ${\cal F}^{-1}$ means the inverse Fourier transform. For all $s$ and $t$, $\Lambda^s$ is an isomorphism between $H^{t}(\R)$ and $H^{t-s}(\R)$ and its inverse is denoted by $\Lambda^{-s}$. In particular, the space $H^s(\R)$ can be seen as $H^s(\R)=\Lambda^{-s}(L^2(\R))$ and $\langle u,v\rangle_{H^s}=\langle\Lambda^s u,\Lambda^s v\rangle_{L^2}$;

\end{itemize}

Although we have a family of operators $\Lambda^s$, $s\in\R$, the most important one for our purposes is $\Lambda^{2}$, which can be identified as the differential operator $1-\p_x^2$, while its inverse is given by $\Lambda^{-2}f=p\ast f$, where $p(x)=e^{-|x|}/2$ and $\ast$ denotes the usual convolution. Instead of $\Lambda^1$, through this paper we will simply use $\Lambda$.

Given an operator $A$, by $\dom(A)$ we mean the domain of $A$. If $A$ and $B$ are two operators with the same domain and range, their commutator is defined by $[A,B]g:=A(B(g))-B(A(g))$. Identifying a function $f$ as the multiplication operator $M_f$, we have $[A,f]\equiv[A,M_f]$, which acts ad the follows: $[A,f]g=A(fg)-f\,A(g)$.

Moreover, we shall make use of the estimates \cite{KP,linares,taylor}:
$\Vert \Lambda^{-2}f \Vert_{H^s}\leq \Vert f\Vert_{H^{s-2}}$,
$\Vert\partial_x f \Vert_{H^{s-1}}\leq \Vert f\Vert_{H^{s}}$ and
$\Vert \partial_x\Lambda^{-2}f \Vert_{H^s}\leq \Vert f\Vert_{H^{s-1}}$.

\begin{lemma}\label{lema2.1}
For $s>1/2$, there is a constant $c_s>0$ such that $\Vert fg \Vert_{H^s}\leq c_s\Vert f\Vert_{H^s}\Vert g\Vert_{H^s}$.
\end{lemma}

\begin{proof}
See \cite{linares}, Theorem 3.5 on page 51, or \cite{taylor}, Exercise 6 on page 320.
\end{proof}

Lemma \ref{lema2.1} is also known as Algebra Property.

\begin{lemma}\label{lema2.2}
If $s>1/2$, then there exists $c_s>0$ such that $\Vert fg \Vert_{H^{s-1}}\leq c_{s}\Vert f\Vert_{H^{s}}\Vert g \Vert_{H^{s-1}}$.
\end{lemma}
\begin{proof}
See Lemma $A1$ in \cite{kato}.
\end{proof}

\begin{lemma}\label{lema2.3}
If $s>1/2$ and $u\in H^{s}(\R)$, then $u$ is bounded and continuous. Moreover, in case we have $s>1/2+k$, then $H^s(\R)\subseteq C^k_0(\R)$.
\end{lemma}

\begin{proof}
See \cite{linares}, theorem 3.2, page 47, or \cite{taylor}, Proposition 1.3, page 317.
\end{proof}

Lemma \ref{lema2.3} is nothing but a Sobolev Embedding Theorem. We observe that if $u\in H^{s}(\R)$, with $s>1/2+k$, for a certain natural number $k$, then $u\in C^{k}_0$ and $\|u\|_{C^k}\leq\|u\|_s$.

\begin{lemma}\label{lema2.4}
Let $m$ be a positive integer and $F\in C^{m+2}(\R)$ be a function such that $F(0)=0$. Then, for every $r\in(1/2,m]$ and $u\in H^r(\R)$, we have $\|F(u)\|_{H^r}\leq \tilde{F}(\|u\|_\infty)\|u\|_{H^r}$, for a certain monotonic and increasing function $\tilde{F}$ depending only on $F$ and $r$.
\end{lemma}

\begin{proof} See \cite{const-mol}.\end{proof}

By Lemma \ref{lema2.4}, if $F$ satisfies its conditions, then $F(u)\in H^r$ for any $u\in H^r,\,r\in(1/2,m]$. Moreover, by the Mean Value Theorem and assuming that $\|u\|_{H^r}$ and $\|v\|_{H^r}$ are bounded, we have $\|F(u)-F(v)\|_{H^r}\leq M\|u-v\|_{H^r}$, for a certain positive constant $M$.

The following diagram is also important to understand some of the demonstrations we shall present in the next section.
\begin{figure}[h!]
\begin{center}
\begin{tikzcd}
H^s(\R)\ni u \arrow{rr}{\p_x}\arrow{rrdd}{\Lambda^{-2}\p_x}  &&\p_x u\in H^{s-1}(\R) \arrow{dd}{\Lambda^{-2}}\\
\\
 &&\Lambda^{-2}\p_x u\in H^{s+1}(\R)
\end{tikzcd}
\end{center}
\caption{Diagram illustrating the composition of the operators $\Lambda^{-2}$ and $\p_x$, which will be useful in Lemma \ref{lema2.9}.}
\end{figure}

In view of ${\bf F1}$, $H^{s+1}(\R)\subseteq H^{s}(\R)$ and then $\|\Lambda^{-2}\p_x u\|_{H^{s+1}}\leq c\|u\|_{H^s}$, for any $u\in H^{s}(\R)$ and a certain constant $c>0$. Let $F$ be a function satisfying Lemma \ref{lema2.4}. Then, not only the diagram holds with $u$ replaced by $F(u)$, $u\in H^{s}(\R)$, $s\in(1/2,m]$, but also 
\bb\label{2.1.1}
\|\Lambda^{-2}\p_x F(u)\|_{H^s}\leq M\|u\|_{ H^s},
\ee
for a certain positive constant $M$.

\subsection{Kato's approach}\label{subsec2.2}

To begin with, let $X$ and $Y$ be two Banach spaces. Consider the problem
\bb\label{2.2.1}
\left\{
\ba{l}
\ds{\f{d u}{dt}+A(u)u=f(u)\in X,\quad t\geq0},\\
\\
u(0)=u_0\in Y,
\ea
\right.
\ee
where $A(u)$ is a linear operator.

In \cite{kato} Kato proved that if certain conditions are satisfied, then the problem \eqref{2.2.1} has a unique solution. 
We are now in position to recall Kato's conditions. The first one is:

\begin{condition}\label{cond2.1} Let $X$ and $Y$ be reflexive Banach spaces, such that $Y\subseteq X$ and the inclusion $Y\hookrightarrow X$ is continuous and dense. In addition, there exists an isomorphism $S:Y\rightarrow X$ such that $\|u\|_Y=\|Su\|_X$.
\end{condition}

We observe that any Hilbert space is reflexive. Moreover, in view of {\bf F1} and {\bf F4}, $X=H^{s}(\R)$, $Y=H^{s-1}(\R)$ and $S=\Lambda$ satisfy Condition \ref{cond2.1}

\begin{condition}\label{cond2.2}
There exist a ball $W$ of radius $R$ such that $0\in W\subseteq Y$ and a family of operators $(A(u))_{u\in W}\subseteq {\cal L}(X)$ such that $-A(u)$ generates a $C_0$ semi-group in $X$ with $\|e^{-s A(u)}\|_{{\cal L}(X)}\leq e^{\be s}$, for any $u\in W$, $s\geq0$, for a certain real number $\be$.
\end{condition}

We recall that if $H$ is a Hilbert space over $\R$ or $\mathbb{C}$, an operator (bounded or not) $A:\dom{(A)}\subseteq H\rightarrow H$ is called $m-$dissipative (in the sense of Philips \cite{phil}, see also \cite{escher}) if and only if $\Re \langle Ax,x\rangle\leq 0$, for all $x\in \dom{(A)}$ (here $\Re$ means the real part of a given complex number), and $\range (\lambda I-A)=H$, for some $\lambda>0$. A densely defined operator $A$ is $m-$dissipative if and only if $A$ and its adjoint $A^\ast$ are dissipative.

\begin{condition}\label{cond2.3}
Let $S$ be the isomorphism in Condition $\ref{cond2.1}$. Then $B(u):=[S,A]S^{-1}\in {\cal L}(X)$. Moreover, there exist constants $c_1$ and $c_2$ such that $\|B(u)\|_{{\cal L}(X)}\leq c_1$, $\|B(u)-B(v)\|_{{\cal L}(x)}\leq c_2\|u-v\|_{Y}$, for all $u,\, v\in W$ 
\end{condition}

\begin{condition}\label{cond2.4}
For any $w\in W$, $Y\subseteq \dom(A(w))$ and $\|A(u)-A(v)\|_{{\cal L}(Y;X)}\leq c_3\|u-v\|_{X}$, for any $u,\,v\in W$.
\end{condition}

\begin{condition}\label{cond2.5}
The function $f:X\rightarrow X$ satisfy the following conditions:
\begin{enumerate}
    \item $\left.f\right|_{W}:W\rightarrow Y$ is bounded, that is, there exists a constant $c_4$ such that $\|f(w)\|_Y\leq c_4$, for all $w\in W$;
    \item $\left.f\right|_{W}:W\rightarrow X$ is Lipschitz when taking the norm of $X$ into account, that is, there is another constant $c_5$ such that $\|f(u)-f(v)\|_{X}\leq c_5\|u-v\|_X$, for all $u,\,v\in W$.
\end{enumerate}
\end{condition}

We would like to observe that the constants mentioned in the conditions above depend on the radius $R$ of $W$, see \cite{escher,kato,mustafa,blanco}.

The following result was proved in \cite{kato} (see Theorem 6), and is the basis to the proof Theorem \ref{teo1.1}.
\begin{lemma}\label{lema2.5}
Consider the problem $(\ref{2.2.1})$ and assume that conditions $\ref{cond2.1}$--$\ref{cond2.5}$ are satisfied. If $u_0\in W$, then there is $T>0$ such that $(\ref{2.2.1})$ has a unique solution $u\in C^0(W,[0,T))\cap C^1(X,[0,T))$, with $u(0)=u_0$.
\end{lemma}

A final, but crucial observation: several evolution equations $u_t=F[u]$ can be seen as a system of the form \eqref{2.2.1}, for each fixed $x$. For this reason, Kato's approach is a useful tool for dealing with these equations, {\it e.g}, see \cite{escher,liu2011,hak,mustafa,blanco}. Equation \eqref{1.0.4}, at first sight, is not eligible to the application of Kato's approach since it is not an evolution equation if $\epsilon\neq0$ and, therefore, not in the form \eqref{2.2.1}. On the other hand, we observe that \eqref{1.0.4} can be rewritten as 
$$(1-\p_x^2)u_t=F[u_{(3)}],$$
where we took $\epsilon=1$ (and this will be enough as we will show in Subsection \ref{subsec2.4}) and $F[u_{3}]$ is the right side of \eqref{1.0.4}. Remembering that the operator $1-\p_x^2$ can be identified with the operator $\Lambda^2$, the last equation can be put in the following form
$$u_t=\Lambda^{-2}F[u_{(3)}],$$
which is nearly in the form \eqref{2.2.1}. We will show very soon, in Subsection \ref{subsec2.4}, that \eqref{1.0.4} can be seen as an equation of the form \eqref{2.2.1}.
\subsection{Auxiliary results}\label{subsec2.3}

In this section we present some technical results needed to prove Theorem \ref{teo1.1}. To prove it, we must split the demonstration in two main cases: $\epsilon=0$ and $\epsilon\neq0$. For the first case, we have the following result:
\begin{lemma}\label{lema2.6}
Let $s>3/2$ and $u_0\in H^s(\R)$. Then, the problem
\bb\label{2.3.1}
\left\{
\ba{l}
u_t+u_{xxx}+\p_x g(u)=0,\quad x\in\R,\quad t\in[0,T),\,T>0,\\
\\
u(x,0)=u_0(x)
\ea\right.
\ee
has a unique solution $u\in C^0([0,T),H^s(\R))\cap C^1([0,T),H^{s-3})$, with $T$ having a lower bound depending only on $\|u_0\|_{H^s}$. Moreover, the map $u_0\mapsto u(\cdot,u_0)$ is continuous in $H^s(\R)$.
\end{lemma}

\begin{proof}
See \cite{katoma}, Theorem I.
\end{proof}

This lemma is actually enough to prove our Theorem \ref{teo1.1} for the case $\epsilon=0$. For the remaining part we need a little more effort to prove it. The next results will play a vital role to this end.

\begin{lemma}\label{lema2.7}
Let  $b$ be a constant, $g\in C^{m+3}(\R)$, with $m\geq2$ and $g(0)=0$, and $u\in H^{s}(\R)$, with $s>3/2$. Then the operator 
\bb\label{2.3.2}
A(u)=(b+g(u))\p_x
\ee
satisfies conditions $\ref{cond2.2}$ and $\ref{cond2.4}$.
\end{lemma}

\begin{lemma}\label{lema2.8}
Let $b$, $g$ and $A(u)$ as in Lemma \ref{lema2.7}. Then the operator $B(u):=[\Lambda,A(u)]\Lambda^{-1}$, with $u\in H^s(\R)$ and $s>3/2$, satisfies condition $\ref{cond2.3}$.
\end{lemma}

The proofs of lemmas \ref{lema2.7} and \ref{lema2.8} can be found in \cite{liu2011} and, therefore, are omitted here, see Lemmas 3.2, 3.3 and 3.4 in the mentioned paper. In the same reference is proven that the function
\bb\label{2.3.3}
f(u):=\Lambda^{-2}(g(u)u_x)+\Lambda^{-2}\p_x\left(bu-h(u)-\f{g'(u)}{2}u_x^2\right)
\ee
satisfies condition \ref{cond2.5}. Although we confirm the result announced there, the demonstration presented in \cite{liu2011} seems to have a small mistake. Then we present a new demonstration, which follows closely the spirit of \cite{liu2011}, but corrects the problem.

In what follows we will employ several different constants, arising from estimates. To avoid a tedious notation, we shall make use of the following convention: we write $\|u\|_X \lesssim \|v\|_Y$ meaning that $\|u\|_X\leq c\|v\|_Y$, for some constant $c>0$.

\begin{lemma}\label{lema2.9}
Assume that $b$, $g$ and $s$ are as in Lemma \ref{lema2.7}, and $h\in C^{m+3}(\R)$, $m\geq 2$, with $h(0)=0$. Then the function in $\eqref{2.3.3}$ satisfies condition \ref{cond2.5}.
\end{lemma}

\begin{proof}
Due to the comments at the end of subsection \ref{subsec2.1}, we can choose $s$ such that if $u\in H^{s}(\R)$, then $f(u)\in H^{s}(\R)$, for a suitable choice of $f$. Let us rewrite $f(u)=f_1(u)+f_2(u)+f_3(u)$, where $f_1(u)=\Lambda^{-2}\p_x\overline{g}(u)$, $f_2(u)=\Lambda^{-2}\p_x(bu-h(u))$ and $f_3(u)=-\Lambda^{-2}\p_x(g'(u)u_x^2/2)$, where $\overline{g}$ is a function such that $\overline{g}'=g$. This condition does not guarantee the existence of a unique function $\overline{g}$, but we can take it unique if we impose the condition $\overline{g}(0)=0$. In particular, this choice makes $\overline{g}$ a function satisfying Lemma \ref{lema2.4}. In addition, since $u\in W$, then $\|u\|_s\leq R$ and $|g'(u)|\leq\sup\{|g'(y)|,\,|y|\leq R\}=: \kappa$. A similar argument also applies to the function $\overline{g}$.

Let us estimate $\|f_i(u)-f_i(v)\|_s$, $i=1, 2, 3$, where $\|\cdot\|_s$ denotes the norm $\|\cdot\|_{H^s}$ for sake of simplicity. We have
$$
\|f_1(u)-f_1(v)\|_s\lesssim\|\overline{g}(u)-\overline{g}(v)\|_s\lesssim\|u-v\|_s.
$$
Also,
$$\|f_2(u)-f_2(v)\|_s\lesssim\|u-v\|_{s-1}+\|h(u)-h(v)\|_{s-1}\lesssim\|u-v\|_{s-1}\lesssim\|u-v\|_{s},$$
and, finally,
$$
\ba{lcl}
\|f_3(u)-f_3(v)\|_s&\lesssim&\|g'(u)(u_x^2-v_x^2)+(g'(u)-g'(v))v_x^2\|_{s-1}\lesssim\|g'(u)\p_x(u+v)\p_x(u-v)\|_{s-1}\\
\\
&&+\|(g'(u)-g'(v))v_x^2\|_{s-1}\lesssim\|\p_x(u+v)\|_{s-1}\|\p_x(u-v)\|_{s-1}+\|v_x^2\|_{s-1}\|u-v\|_{s-1},
\ea
$$
where we used Lemma \ref{lema2.2} and the fact that $\|g'(u)-g'(v)\|_s\leq M\|u-v\|_s$, for some constant $M>0$. Therefore, $\|f_3(u)-f_3(v)\|_s\lesssim\|u-v\|_s$. As a consequence of these facts, we conclude that $\|f(u)-f(v)\|_s\lesssim\|u-v\|_s$. This proves part 3 of condition \ref{cond2.5}.

Now, let $u\in H^{s-1}(\R)$. Similarly to what we have done for $u\in H^{s}(\R)$, we easily conclude that $f(u)\in H^{s-1}$, which proves part 1 of condition \ref{cond2.5}. Moreover, the proof that $\|f(u)-f(v)\|_{s-1}\lesssim\|u-v\|_{s-1}$ is very similar to the previous one and we omit its demonstration. Noting that $f(0)=0$, the previous inequality, jointly with the fact that $u\in W$, implies part 2 of the condition \ref{cond2.5}.
\end{proof}

\subsection{Proof of Theorem \ref{teo1.2} and Corollary \ref{cor1.1}}\label{subsec2.4}

Let $b:=g(0)$, $G(u):=g(u)-b$ and $m:=u-u_{xx}$. Then $G$ satisfies the conditions in Lemma \ref{lema2.4}.

We note that
$$
\ba{lcl}
\ds{m_t+\p_x h(u)-\p_{x}\left(\f{g'(u)}{2}u_x^2+g(u)u_{xx}\right)}&=&\ds{m_t+\p_x h(u)-bu_{xxx}-\p_{x}\left(\f{G'(u)}{2}u_x^2+G(u)u_{xx}\right)}\\
\\
&=&\ds{m_t+\Lambda^2\left(bu_x+G(u)u_x\right)-G(u)u_x}\\
\\
&&\ds{-\p_x\left(bu-h(u)-\f{G'(u)}{2}u_x^2\right).}
\ea
$$

Application of the operator $\Lambda^{-2}$ to both sides of the last equation enable us to consider the problem
\bb\label{2.4.1}
\left\{
\ba{l}
\ds{u_t+(b+G(u))u_x=\Lambda^{-2}\left(G(u)u_x+\p_x\left(bu-h(u)-\f{G'(u)}{2}u_x^2\right)\right)},\\
\\
u(x,0)=u_0(x).
\ea
\right.
\ee

Note that (\ref{2.4.1}) is equivalent to (\ref{1.0.6}) and it is, therefore, enough to prove the well- posedness for \eqref{2.4.1}. We now observe that \eqref{2.4.1} is of the form \eqref{2.2.1}, with
\bb\label{2.4.2}
A:=(b+G(u))\p_x,\quad f(u):=\Lambda^{-2}(G(u)u_x)+\Lambda^{-2}\p_x\left(bu-h(u)-\f{G'(u)}{2}u_x^2\right),
\ee
and $h$ and $G$ satisfy the conditions in Lemma \ref{lema2.4}. This implies that $A$ and $f$ satisfy lemmas \ref{lema2.7}--\ref{lema2.9}. Then, by Lemma \ref{lema2.5} we have granted the existence and uniqueness of solutions to (\ref{2.4.1}), which implies Theorem \ref{teo1.2}.

Corollary \ref{cor1.1} can be demonstrated by replacing $g$ by $g-\Gamma$ and applying Theorem \ref{teo1.2} to the shifted function.

\subsection{Consequences of Theorem \ref{teo1.2} and proof of Theorem \ref{teo1.1}}\label{subsec2.5}

As consequences of Theorem \ref{teo1.2}, we have the following two corollaries:
\begin{corollary}{\bf (Liu, Yin, \cite{liu2011})}\label{cor2.1}
Assume that $h,\,g\in C^{m+3}(\R)$, $m\geq 2$ and $h(0)=g(0)=0$. Given $u_0\in H^{s}(\R)$, $3/2<s<m$, there exists a maximal time $T=T(u_0)>0$, and a unique solution $u$ to $\eqref{1.0.7}$ such that $u=u(\cdot,u_0)\in C^0(H^s(\R);[0,T))\cap C^{1}(H^s(\R),[0,T))$, continuously dependent on the initial data.
\end{corollary}

We observe that in Theorem \ref{teo1.2} we removed the limiting condition $g(0)=0$, which makes an improvement in the results in \cite{liu2011}. Another consequence of Theorem \ref{teo1.2} is (below $W$ is a ball in $H^{s-1}$):
\begin{corollary}\label{cor2.2}{\bf (Mustafa, \cite{mustafa})}
Assume that $u_0 \in W$. Then the problem $\eqref{1.0.7}$, with $\Gamma=0$ and $g(u)=u$ has a unique solution in $C^0(W,[0, T])\cap C^1(H^{s-1}(\R),[0, T]]$ such that $u(x,0) = u_0(x)$ for a certain $T >0$, and the solution depends continuously on the initial data.\end{corollary}

Now we prove Theorem \ref{teo1.1}. To do it, we only need to consider the cases $\epsilon=0$ and $\epsilon\neq0$. The first case is an immediate consequence of Lemma \ref{lema2.6}, as we have already pointed out in the comment after Lemma \ref{lema2.6}. Regarding the case $\epsilon\neq0$, let us take the global diffeomorphism $(x,t,u)\mapsto(x/ \epsilon,t/\epsilon,u)$, which makes possible to rewrite the equation (\ref{1.0.5}) (or \eqref{1.0.4}) as
$$
u_t- u_{txx}=uu_{xxx}+2 u_{x}u_{xx}+(\al-3u+\be u^2+\gamma u^3)u_x+\f{\Gamma}{\epsilon^2} u_{xxx},
$$
which is nothing but \eqref{1.0.7} with $\Gamma$ replaced by $\Gamma/\epsilon^2$ and
$$h(u)=-\alpha u+\f{3}{2}u^2-\f{\beta}{3}u^3-\f{\gamma}{4}u^4\quad\text{and}\quad g(u)=u.$$

Then Theorem \ref{teo1.1} is an immediate consequence of Theorem \ref{teo1.2}.

\section{Conservation laws}\label{sec3}

Let $x$ and $t$ be independent variables and $u=u(x,t)$ be a field variable. A smooth function $P$ depending on $x$, $t$, $u$ and derivatives of $u$ with respect to the independent variables up to a finite, but unspecified, order is called {\it differential function}. We shall denote by $P[u]$ a general differential function and, when its order play some relevance, we simply write $P[u_{(n)}]$, meaning that $P$ is a differential function up to $n-$th order. For further details, see \cite{olverbook}, page 288.

We observe that any $n-$th order differential equation in two independent variables $t$ and $x$ can be generically written as $E[u_{(n)}]=0$.

A conservation law for an equation $E[u_{(n)}]=0$, with two independent variables $(x,t)$ and a field variable $u=u(x,t)$, is a divergence expression
\bb\label{3.0.1}
D_tC^{0}+D_{x}C^1=0\quad\text{mod}\quad E[u_{(n)}]=0.
\ee

The sentence expressed in (\ref{3.0.1}) should be understood as follows: the divergence $D_tC^0+D_xC^1$ does not necessarily need to be $0$ everywhere, but it must vanish on the set (manifold) determined by $E[u_{(n)}]=0$ and on all of its differential consequences.

In (\ref{3.0.1}), 
$$
D_t=\f{\p}{\p t}+u_{t}\f{\p}{\p u}+u_{tx}\f{\p}{\p u_x}+u_{tt}\f{\p}{\p u_t}+\cdots,\quad\text{and}\quad
D_x=\f{\p}{\p x}+u_{x}\f{\p}{\p u}+u_{xx}\f{\p}{\p u_x}+u_{xt}\f{\p}{\p u_t}+\cdots
$$
are the total derivative operators with respect to $t$ and $x$, respectively. For further details, see \cite{olverbook}, chapter 5.

The pair $(C^0,C^1)$ satisfying (\ref{3.0.1}) is called {\it conserved current} (of the equation $E[u_{(n)}]=0$) and both of them are differential functions. We are in position to make three important observations.

\begin{remark}\label{rem3.1}
While the operator $d/dt$ is a vector field (from a geometrical point of view), the operator $D_t$ above is a contact distribution on the $r-$jet (for some $r$) $J^r(\R,n)$. Roughly speaking, the final result obtained after applying both operators is the same. However, they are conceptually different. For further details, see \cite{manno}.

\end{remark}

\begin{remark}\label{rem3.2}
Under the very mild hypothesis that $C^0$ and $C^1$ are continuous, we can integrate equation $(\ref{3.0.1})$ over a domain $\Omega=(a,b)\subseteq\R$ and interchange the total derivation and the integral over the domain, which gives
\bb\label{3.0.2}
\f{d}{d t}\int_{\Omega}C^0\,dx=-\left.C^1\right|_{a}^b.
\ee
Here we indeed allow $a=-\infty$ and $b=\infty$. In $(\ref{3.0.2})$, $C^1$ is the flux across the boundary and $C^0$ is the conserved density. Whenever the condition $\left.C^1\right|_{a}^b=0$ is satisfied to a certain solution $u$ of the equation $E[u]=0$, we conclude that the functional 
$$u\rightarrow J[u]:=\int_{\Omega}C^0\,dx$$
is independent of $t$.
\end{remark}

\begin{remark}\label{rem3.3}
Condition $(\ref{3.0.1})$ can be rewritten as
$$
D_tC^{0}+D_{x}C^1=Q[u]E[u].
$$
Above, the differential function $Q$ is called {\it characteristic of the conservation law}, see \cite{olverbook}, chapter 5.

Let 
$$E_u=\f{\p}{\p u}-D_t\f{\p}{\p u_t}-D_x\f{\p}{\p u_x}+\cdots$$ be the Euler-Lagrange operator. From Theorem 4.7 of \cite{olverbook}, we know that $E_u(L)=0$ if and only if there exist differential functions $P^0$ and $P^1$ such that $L=D_t P^0+D_x P^1$. Therefore, this result and the above equation read
$$
E_u(Q[u]E[u])=0.
$$
Moreover, we observe that if $Q_1[u]$ and $Q_2[u]$ are characteristics of conservation laws of $E[u]=0$, then $Q=\al Q_1[u]+\be Q_2[u]$, for any scalars $\al$ and $\be$, is also a characteristic of a certain conservation law of the same equation. This can be seen from the following fact: the linear combination of any conserved current of a given equation is still a conserved current of the same equation, see, {\it e.g.}, \cite{popPLA}.
\end{remark}

\begin{theorem}\label{teo3.1}\label{teo3.2}
 Let $m=u-\epsilon^2 u_{xx}$, $(C^0,C^1)$ a conserved current for equation $(\ref{1.0.4})$, with $\epsilon\neq0$, and $Q[u]$ a  characteristic up to second order. Then $Q=c_1+c_2u$, where $c_1$ and $c_2$ are arbitrary constants, for any values of the constants in $(\ref{1.0.4})$. In the particular case where $\be=\gamma=0$ and $\Gamma=-\al\epsilon^2$, we have a third characteristic given by $Q[u]=m^{-1/2}$.

Furthermore, the components $C^0$ and $C^1$ corresponding to the characteristics are:
\begin{enumerate}
    \item For the characteristic $Q=1$ we have the components
    $$
    C^0=u\quad\text{and}\quad
     C^1=\f{3}{2}u^2-\epsilon^2u_{tx}-\epsilon^2uu_{xx}-\f{\epsilon^2}{2}u_x^2-\al u-\f{\be}{3}u^3-\f{\gamma}{4}u^4-\Gamma u_{xx}.
     $$
      
    \item For the characteristic $Q=u$ we have the components
    
    $$C^0=\f{u^2+\epsilon^2 u_{x}^2}{2} \quad\text{and}\quad
    C^1=u^3-\epsilon^2 u^2u_{xx}-\epsilon^2uu_{tx}+\Gamma \f{u_x^2}{2}-\Gamma uu_{xx}-\al\f{u^2}{2}-\be\f{u^4}{4}-\f{\gamma}{5}u^5.$$
    
    \item For the characteristic $Q=\frac{1}{2}(u-\epsilon^2 u_{xx})^{-1/2}$ we have the components
    $$C^0=\sqrt{m}\quad\text{and}\quad  C^1=(u-\al)\sqrt{m},$$
    where $m=u-\epsilon^2 u_{xx}$, $\beta=\gamma=0$ and $\Gamma = -\alpha\epsilon^2$.
\end{enumerate}
\end{theorem}

\begin{proof}
Let $\Delta=m_t+u\,m_x+2u_xm-\al u_x-\be u^2u_x-\gamma u^3u_x-\Gamma u_{xxx}$ and $Q=Q[u_{(2)}]$. From Remark \ref{rem3.3} and the condition
$$E_u(Q[u_{(2)}]\Delta)=0$$
we obtain a system of differential equations to $Q$, whose solution is $Q=c_1+c_2u$, under no restrictions on the parameters, and $Q=(u-\epsilon^2u_{xx})^{-1/2}$ provided that $\be=\gamma=0$ and $\Gamma=-\al\epsilon^2$.

For the explicit form of the components $C^0$ and $C^1$, it is enough to multiply $\Delta$ by the respective characteristics and manipulate the resulting expression.
\end{proof}

      
%
    

\begin{remark}\label{rem3.4}
We note that the last conserved current is {\it formal}. However, it is {\it truly} a conserved current whenever $m$ is non-negative/non-positive, eventually replacing $m$ by $-m$ in case $m\leq0$. The same observation is also applied to the characteristic $m^{-1/2}$ in Theorem $\ref{teo3.1}$.
\end{remark}

\begin{remark}\label{rem3.5}
The fact that $m$ is non-negative/non-positive is of importance to prove global properties of the solutions of the CH equation, {\it e.g}, see Theorem 7 of \cite{escher}, or Theorem 4.1 of \cite{blanco}. We will also explore similar facts in Section \ref{sec4}.
\end{remark}

\begin{remark}\label{rem3.6}
The classification of conservation laws for the case $\epsilon=0$ is richer than the case considered in Theorem $\ref{teo3.1}$. Actually, if $\epsilon=\be=\gamma=0$, we have the KdV equation, which is known to have an infinite hierarchy of conservation laws \cite{miura}.  Moreover, even the case $\epsilon\neq0$ is very rich, since we then have the CH equation when $\al=\be=\gamma=\Gamma=0$, which has an infinite hierarchy of conservation laws, see \cite{chprl} and \cite{lenjpa}.
\end{remark}

\section{Properties of solutions derived from the conservation laws}\label{sec4}

Given a function $u=u(x,t)$, it is sometimes natural to ask whether the function $u(\cdot,t)$ belongs to $\in H^1(\R)$, for $t\in[0,\infty)$ fixed. Moreover, in case this fact holds to each value of $t$, we say that $u\in H^{1}(\R)$.

\begin{theorem}\label{teo4.1}
Let $u$ be a solution of $(\ref{1.0.4})$, with $\epsilon\neq0$, satisfying $u(x,0)=u_{0}(x),\,u_0\in H^1(\R)$, and such that $u(x,t),\,u_{x}(x,t)\rightarrow0$ as $x\rightarrow\pm\infty$ and whose second derivatives are bounded on the entire real line, for any $t\in[0,\infty)$. Then $u\in H^1(\R)$.
\end{theorem}

Before proving Theorem \ref{teo4.1}, the following observation is necessary: if $\epsilon\neq0$, then
$$\int_{\R}\left(u^2(x,t) + \epsilon^2 u_x^2(x,t)\right)dx = \f{1}{|\epsilon|}\int_{\R}\left(u^2(x/\epsilon,t) + u_x^2(x/\epsilon,t)\right)dx.$$
It means that we can identify the space of functions $u$ such that $\int_{\R}\left(u^2(x,t) + \epsilon^2 u_x^2(x,t)\right)dx<\infty$ with $H^1(\R)$. Moreover, we have indeed
$$\min\{1,\epsilon^2\}\int_{\R}\left(u^2(x,t) + u_x^2(x,t)\right)dx\leq \int_{\R}\left(u^2(x,t) + \epsilon^2 u_x^2(x,t)\right)dx \leq \max\{1,\epsilon^2\} \int_{\R}\left(u^2(x,t) + u_x^2(x,t)\right)dx,$$ and we can now proceed with the proof of Theorem \ref{teo4.1}.

\begin{proof}
Let us first define
$$
J[u]=\f{1}{2}\int_{\R}(u^2+\epsilon^2u_{x}^2)dx\quad\text{and}\quad J[u_0]=\f{1}{2}\int_{\R}(u^2_0+\epsilon^ 2 u_{0x}^2)dx,\quad u_{0x}:=u_x(x,0)=u_{0}'(x).
$$
We observe that if we prove that $J[u]$ is constant, then we automatically show that $u(\cdot,t)\in H^1(\R)$, for any $t$, since $\|u\|^2_{H^{1}}=2 J[u]$.

Integrating the conservation law obtained from the characteristic $Q=u$, we have
$$
\ba{lcl}
\ds{\f{d}{dt}J[u]}&=&\ds{\int_{\R}D_t\left(\f{u^2+\epsilon^2 u_{x}^2}{2}\right)}\\
\\
&=&-
\ds{\left.\left(u^3-\epsilon^2 u^2u_{xx}-\epsilon^2uu_{tx}+\Gamma \f{u_x^2}{2}-\Gamma uu_{xx}-\al\f{u^2}{2}-\be\f{u^4}{4}-\f{\gamma}{5}u^5\right)\right|_{-\infty}^{+\infty}=0}.
\ea
$$

This implies that $J[u]=c$ and, at $t=0$, we have $c=J[u_0]$, meaning that $J[u]=J[u_0]$, for all $t$.
\end{proof}

\begin{theorem}\label{teo4.2}
Let u=u(x,t) be a solution of $(\ref{1.0.4})$, $u_0(x):=u(x,0)$, such that $u(\cdot,t),\,u_x(\cdot,t)$ and $u_{xx}(\cdot,t)$ are integrable and vanishing at $\pm\infty$, for all $t\in[0,\infty)$. Let $m=u-\epsilon^2 u_{xx}$ and $m_0:=u_0-\epsilon^2 u_0''$. Then
$$\int_\R m\,dx=\int_\R m_0\,dx.$$
\end{theorem}

\begin{proof}
Let us consider the first conserved current in Theorem \ref{teo3.2}. We have
\bb\label{3.1.1}
\ba{lcl}
0&=&\ds{ D_t(u)+D_x\left(\f{3}{2}u^2-\epsilon^2u_{tx}-\epsilon^2uu_{xx}-\f{\epsilon^2}{2}u_x^2-\al u-\f{\be}{3}u^3-\f{\gamma}{4}u^4-\Gamma u_{xx}\right)}\\
\\
&=&\ds{D_t\left(u -\epsilon^2 u_{xx}\right)+D_x\left(\f{3}{2}u^2-\epsilon^2uu_{xx}-\f{\epsilon^2}{2}u_x^2-\al u-\f{\be}{3}u^3-\f{\gamma}{4}u^4-\Gamma u_{xx}\right)}\\
\\
&=:&\ds{D_tm+D_x\tilde{C}},
\ea
\ee
after transferring the term $\epsilon^2u_{tx}$ from the second part of the first equation to the first part of the second equation above. We observe that $\tilde{C}\rightarrow0$ as $|x|\rightarrow\infty$ with our assumptions.

Define
$$A(t):=\int_{\R}m(x,t)\,dx.$$
Then
$$A(0)=\int_{\R}m(x,0)\,dx=\int_{\R}m_0\,dx$$
and
$$
\f{d A}{dt}=\int_{\R}D_tm\,dx=\left.\tilde{C}\right|_{-\infty}^{+\infty}=0,
$$
which means that $A(t)=A(0)$, for all $t$.
\end{proof}

We observe that both theorems \ref{teo4.1} and \ref{teo4.2} are valid for $\epsilon=0$. In the case of Theorem \ref{teo4.2}, $H^1(\R)$ is replaced by $L^2(\R)$. More interestingly, although these results were already known for the CH equation, they were proven previously using slightly different methodologies and conditions, 
after proving the existence of solutions to a initial problem involving (\ref{1.0.4}) and using that $m_0$ does not change its sign. Thus, in that case one can use the monotone convergence theorem, and Lebesgue's dominated convergence theorem. See, for instance, \cite{const1998-1}.

\begin{corollary}\label{cor4.1}
Under the conditions of Theorem $\ref{teo4.2}$, 
$$\int_\R u\,dx=\int_{\R}m\,dx.$$
\end{corollary}
\begin{proof}
Integrating equation (\ref{3.1.1}) over $\R$, we conclude that
$$
\f{d}{dt}\int_\R u\,dx=\f{d}{dt}\int_R m\,dx=0.
$$
Noticing that $u=m+\epsilon^2 u_{xx}$ and remembering that $u$ and its derivatives vanish at infinity, the last equality and Theorem \ref{teo4.2} yields
$$
\int_\R u\,dx=\int_\R u_0(x)\,dx=\int_{\R}(m_0+\epsilon^2 u_{xx}(x,0))\,dx=\int_{\R}m_0\,dx=\int_{\R}m\,dx.
$$
\end{proof}

In what follows, $\sign{x}$ denotes the sign function, that is, $\sign{(x)}=+1$, if $x>0$, and $\sign{(x)}=-1$, if $x<0$.

\begin{corollary}\label{cor4.2}
Let $u$ be a solution of $(\ref{1.0.4})$ in $\R\times[0,T)$, for a certain $T>0$, $u_0(x):=u(x,0)$, $m:=u-\epsilon^2 u_{xx}$, $m_0=u_0-\epsilon^2 u_0''$. Assume that $m_0\in L^1(\R)\cap H^1(\R)$, its sign does not change and $\sign{(m)}=\sign{(m_0)}$, for any $(x,t)$. Then
\begin{enumerate}
    \item $\sign{(u)}=\sign{(u_0)}$ and they do not change;
    \item $-\epsilon^2 u_{x}(x,t)\leq\|m_0\|_{L^1}$, for any $(x,t)\in\R\times[0,T)$.
\end{enumerate}
\end{corollary}
\begin{proof}
The first part follows from the fact that $u=p\ast m$, where $p(x)=e^{-|x|}/2$. Since $p>0$, then $\sign{(u)}=\sign{(m)}$, for any $(x,t)\in\R\times[0,T)$.

To prove the second part, let us first assume $m_0\geq0$. By Theorem \ref{teo4.2} we have
$$
\|m_0\|_{L^1(\R)}=\int_\R m_0\,dx=\int_{\R}m\,dx\geq\int_{-\infty}^xm\,dx=\left(\int_{-\infty}^x u\right)-\epsilon^2 u_x(x,t). 
$$
Since $u\geq0$ and
$$
0\leq\int_{-\infty}^xu\,dx\leq\int_\R u\,dx<\infty,
$$
we have $\|m_0\|_{L^1(\R)}\geq -\epsilon^2 u_{x}(x,t)$. 

Let us now prove the inequality whenever $m_0\leq0$. In this case, we have $-m\geq 0$ and $-u\geq 0$ as well. Then
$$
-\int_{\R} u dx = -\int_{\R} m dx = -\int_{\R}m_0dx < \infty.
$$
Since
$$0\geq \int_{-\infty}^xmdx = \int_{-\infty}^x(u-\epsilon^2 u_{xx})dx = \int_{-\infty}^x(u-\epsilon^2u_{xx})dx,$$
we have
$$
-\epsilon^2 u_x(x,t)\leq -\int_{-\infty}^x u\,dx\leq -\int_\R m_0\,dx.
$$
As a consequence we have
$$
-\epsilon^2 u_x(x,t)\leq -\int_{-\infty}^x u\,dx\leq \int_\R -u\,dx=\int_\R -m_0\,dx=\|m_0\|_{L^1(\R)}. 
$$
\end{proof}

\section{Classification of bounded travelling wave solutions}\label{sec5}

Here we proceed  with a qualitative analysis of the travelling wave solutions of equation (\ref{1.0.4}). We shall closely follow the ideas developed by Lenells in \cite{len} for the Camassa-Holm equation. Let us first recall that for $\al=\be=\gamma=\Gamma=0$ and $\epsilon=1$, the classification of bounded traveling wave solutions was made in \cite{len}. For $\be=\gamma=0$ corresponding to the Dullin-Gottwald-Holm equation, the analysis was recently considered by the first author in \cite{raspa}. Therefore, in what follows, we will only consider $(\gamma,\beta)\neq (0,0)$.

Our strategy to classify the bounded travelling waves for \eqref{1.0.4} is the following: firstly we use the first conservation law in Theorem \ref{teo3.2} to obtain a quadrature to equation \eqref{1.0.4}. Then we carry out a classification of the waves following the ideas of \cite{len}, see also \cite{raspa}.

\subsection{Quadrature of \eqref{1.0.4}}\label{subsec5.1}
Let us now consider the conservation law obtained from the first conserved current in Theorem \ref{teo4.2}, that is,
\bb\label{5.1.1}
D_t(u)+D_x\left(\f{3}{2}u^2-\epsilon^2u_{tx}-\epsilon^2uu_{xx}-\f{\epsilon^2}{2}u_x^2-\al u-\f{\be}{3}u^3-\f{\gamma}{4}u^4-\Gamma u_{xx}\right)=0.
\ee
Let $z:=x-ct$, $c\neq0$, $u=\phi(z)$. Equation (\ref{5.1.1}), after integration, reads:
\bb\label{5.1.2}
-c\phi+\f{3}{2}\phi^2-\epsilon^2\phi\phi''-\f{\epsilon^2}{2}(\phi')^2+c\epsilon^2\phi''-\al\phi-\f{\be}{3}\phi^3-\f{\gamma}{4}\phi^4-\Gamma\phi''=A,
\ee
where $A$ is a constant of integration. Multiplying the latter equation by $\phi'$ and integrating again, we have:
\bb\label{5.1.3}
(\phi')^2(\epsilon^2(c-\phi)-\Gamma)=B+A\phi+(c+\al)\phi^2-\phi^3+\f{\be}{6}\phi^4+\f{\gamma}{10}\phi^5,
\ee
where $B$ is another constant of integration. Renaming the constants, we can rewrite (\ref{5.1.3}) in the following more convenient form
\bb\label{5.1.4}
(\phi')^2=\f{P(\phi)}{\epsilon^2(c-\phi)-\Gamma},\quad\text{with}\quad P(\phi)=B+2A\phi+(c+\al)\phi^2-\phi^3+\f{\be}{6}\phi^4+\f{\gamma}{10}\phi^5.
\ee

\subsection{Preliminaries and types of waves}\label{subsec5.2}

Before proceeding with the results, we shall consider some aspects of the theory that will be used here. In some parts of the text we follow closely the presentation in \cite{raspa,len,lenjmaa}.

We begin with notation: $z\uparrow z_0$ means that $z<z_0$ and $z\rightarrow z_0$, whereas $z\downarrow z_0$ means $z\rightarrow z_0$, but $z>z_0$. The space $H^1_{loc}(\R)$ is consisted of functions $u$ which belong to the space $H^1(K)$ for every compact subset $K\subset \R$.

\begin{definition}
A function $\phi(z)\in H^1_{loc}(\R)$, where $z=x-ct$ and $c$ denotes the wave speed of $\phi$, is said to be a traveling wave solution for \eqref{1.0.4} if it satisfies \eqref{5.1.2} in the sense of distributions.
\end{definition}

In \cite{len}, the author classified the existence of traveling wave solutions of the quadrature form
\begin{align}\label{5.2.1}
    (\phi')^2= F(\phi),
\end{align}
where $F$ represents a rational function of $\phi$, based on a qualitative analysis of real zeros and poles of the function $F$. The discussion presented in \cite{len} can be summarized as the following:
\begin{enumerate}
    \item If $F$ has only a simple zero at $z_1$ and $F(\phi)>0$ for $z_1<\phi$, then no bounded traveling wave solutions will exist for $\phi>z_1$.
    \item If $F$ has two simple zeros $z_1$ and $z_2$ and $F(\phi)>0$ for $z_1<\phi<z_2$, then there exists a smooth periodic traveling wave solution $\phi$ of \eqref{5.2.1} with $z_1=\min\limits_{z\in\R} \phi(z)$ and $z_2=\max\limits_{z\in\R} \phi(z)$.
    \item If $F$ has a double zero $z_1$, a simple zero $z_2$ and $F(\phi)>0$ for $z_1<\phi<z_2$, there exists a smooth solution $\phi$ of \eqref{5.2.1} with $z_1=\inf\limits_{z\in\R}\phi(z)$, $z_2=\max\limits_{z\in\R} \phi(z)$ and $\phi\downarrow z_1$ as $z\to\pm\infty$.
\end{enumerate}

In terms of weak solutions, the following discussion can be drawn for the poles of $F$:
\begin{enumerate}
    \item[4.] Peakons will exist when $\phi$ satisfies \eqref{5.2.1} and the pole $a$ of $F$ is removable:
    \begin{align*}
        0 \neq \lim\limits_{z\uparrow z_0} \phi'(z) = - \lim\limits_{z\downarrow z_0} \phi'(z) \neq \pm \infty,
    \end{align*}
    where $z_0\in \R$ is such that $\phi(z_0)=a$.
    \item[5.] Cuspons will exist when $\phi$ satisfies \eqref{5.2.1} and the pole $a=\min\limits_{z\in\R} \phi(z)$ (or taken as the maximum) of $F$ is non-removable:
    \begin{align*}
        \lim\limits_{z\uparrow z_0} \phi'(z) = - \lim\limits_{z\downarrow z_0} \phi'(z) = \pm \infty.
    \end{align*}
\end{enumerate}

Regarding the existence of weak solutions, we observe that it is necessary to analyse their behaviour before and after points $z_0$ such that $\phi(z_0)=a$, where $a$ denotes a pole of $F$. As we will see for the case $\epsilon\neq 0$, weak solutions will lose differentiability only on $z_0$. Consequently, apart from the fact that a zero can remove the singularity of $F$ and lead to peakons or cuspons, the zeros of $F$ will only tell the smooth qualitative behaviour (e.g. periodic, with decay) of $\phi$ away from $z_0$, while the weak feature of solutions is determined by the existence of poles in $F$.

 Before proceeding with the classification results, we observe that any constant function $u(x,t)=u_0$ is a solution of \eqref{1.0.4} for any choice of the parameters and for this reason we make the weak assumption that solutions found in the classifications are not constant.

Consider equation (\ref{5.1.2}). Since the analysis of existence of bounded traveling waves  will be purely qualitative, from now on we shall conveniently take the scalings $\gamma\mapsto 10\gamma$ and $\beta\mapsto 6\beta$ when necessary to write
\begin{align}\label{ast}
    -2c\phi + 3\phi^2-2\epsilon^2\phi \phi''-\epsilon^2(\phi')^2+2c \epsilon^2\phi''-2\alpha \phi -4\beta\phi^3-5\gamma\phi^4-2\Gamma \phi'' = 2A
\end{align}
and the quadrature \eqref{5.1.4} in the form
\begin{align}\label{5.2.2}
    (\phi')^2 = \f{P(\phi)}{\epsilon^2(c-\phi)-\Gamma}, \quad P(\phi) = \gamma \phi^5 + \beta \phi^4 - \phi^3 +(c+\alpha)\phi^2+2A\phi + B.
\end{align}
If $\phi$ is a bounded smooth solution of (\ref{5.2.2}), let
\begin{align*}
    m=\inf\limits_{z\in\R}\phi(z),\quad M=\sup\limits_{z\in\R}\phi(z).
\end{align*}
Then we use the fact that $\phi$ is continuous and $\phi'\to 0$ as $z\to m$ or $z\to M$ to obtain that the infimum and supremum of smooth solutions $\phi$ are zeros of $P(\phi)$. At points where $\phi=(\epsilon^2 c - \Gamma)/\epsilon^2$, the behaviour of infimum and supremum may change since $\phi'$ blows up.

According to the theory briefly discussed, we must analyse the zeros and the sign of $P(\phi)$ based on the placement of $\phi$ among the zeros. However, it is important to observe that, from the quadrature form \eqref{5.1.4}, we must guarantee that the condition
$$F(\phi) = \f{P(\phi)}{\epsilon^2(c-\phi)-\Gamma}>0$$ holds. However, whenever $\epsilon=0,$ the pole of $F(\phi)$ is removed, no weak solutions will exist and we can proceed with the classification.

In what follows, we separately prove the classification for $\epsilon=0$ and $\epsilon\neq0$.

\subsection{Case $\epsilon=0$}

The choice $\epsilon=0$ leads to the quadrature
\begin{align*}
    (\phi')^2 = -\f{P(\phi)}{\Gamma},
\end{align*}
where $P(\phi)$ is given as in \eqref{5.2.2}, poles are removed and no bounded weak traveling wave solutions will exist. Observe that this quadrature makes sense since we assumed $(\epsilon,\Gamma)\neq (0,0)$.

Now consider the polynomial $P(\phi)$ and observe that it will have one, three or five real zeros if $\gamma\neq 0$ and none, two or four real zeros if $\gamma=0$. Before proceeding with our classification results, we shall take a look into compatibility conditions according to the number of zeros of each case.

Firstly assume $\gamma\neq 0$ and suppose that $P(\phi)$ has three real zeros $r_1\leq r_2 \leq r_3$ counted without their multiplicities and a complex zero $z_0$ in such a way we can write
\begin{align*}
    \gamma \phi^5 + \beta \phi^4 - \phi^3 + (\alpha+c)\phi^2+ 2A\phi + B = \gamma(\phi-r_1)(\phi-r_2)(\phi-r_3)|\phi-z_0|^2.
\end{align*}
Comparing the coefficients of $\phi^4,\phi^3,\phi^2$ and $\phi$, the compatibility conditions will be given by
\begin{align}\label{5.3.1}
    \begin{aligned}
    2 \text{Re}(z_0) &= -\f{\beta+\gamma(r_1+r_2+r_3)}{\gamma},\\
    \gamma(r_1^2+&r_2^2+r_3^2)+\beta(r_1+r_2+r_3) + \gamma(r_1r_2+r_1r_3+r_2r_3) -\gamma |z_0|^2=1,\\
    \gamma|z_0|^2(&r_1r_2+r_1r_3+r_2r_3) - \gamma r_1r_2r_3(r_1+r_2+r_3)-\beta r_1r_2r_3=2A, 
    \end{aligned}
\end{align}
while the constant term yields
\begin{align}\label{5.3.2}
    B= -\gamma r_1r_2r_3\vert z_0\vert^2.
\end{align}
For the sake of classification, two cases will arise:
\begin{enumerate}
    \item all real zeros are simple: $r_1<r_2<r_3$. In this case, we can have $r_1<\phi< r_2<r_3$ or $r_1<r_2<\phi< r_3$.
    \begin{itemize}
        \item If $r_1\leq\phi\leq r_2<r_3$, then $-P(\phi)/\Gamma >0$ if and only if $\tilde{\gamma}:=\gamma/\Gamma <0$. For this case we know that there will exist a smooth periodic travelling wave solution $\phi$ with $r_1=\min\limits_{z\in\R}\phi(z)$ and $r_2=\max\limits_{z\in\R}\phi(z)$.
        
        \item If $r_1<r_2\leq\phi\leq r_3$, then $-P(\phi)/\Gamma >0$ if and only if $\tilde{\gamma}:=\gamma/\Gamma >0$. In this case there will exist a smooth periodic travelling wave solution $\phi$ with $r_2=\min\limits_{z\in\R}\phi(z)$ and $r_3=\max\limits_{z\in\R}\phi(z)$, finishing the case of three simple zeros.
    \end{itemize}
    \item only one real zero is double: $r_1=r_2<r_3$ or $r_1<r_2=r_3$. The possibilities for $\phi$ are now $r_1=r_2<\phi<r_3$ or $r_1<\phi<r_2=r_3$.
    \begin{itemize}
        \item If $r_1=r_2<\phi\leq r_3$, then $-P(\phi)/\Gamma >0$ if and only if $\tilde{\gamma}>0$ and there will exist a smooth travelling wave solution $\phi$ with $r_2=\inf\limits_{z\in\R}\phi(z)$, $r_3=\max\limits_{z\in\R}\phi(z)$ and $\phi\downarrow r_2$ as $z\to\pm\infty$.
        
        \item If $r_1\leq \phi<r_2=r_3$, then $-P(\phi)/\Gamma >0$ if and only if $\tilde{\gamma}<0$ and there will exist a smooth travelling wave solution $\phi$ with $r_1=\min\limits_{z\in\R}\phi(z)$, $r_2=\sup\limits_{z\in\R}\phi(z)$ and $\phi\uparrow r_2$ as $z\to\pm\infty$.
    \end{itemize}
\end{enumerate}

These two cases prove the following theorem:
\begin{theorem}\label{teo5.1}\textbf{(Case $\gamma\neq 0$ with three real zeros)}
Let $\gamma\neq 0$, $\tilde{\gamma}=\gamma/\Gamma$ and suppose $r_1\leq r_2\leq r_3\in\R$ and $z_0\in\mathbb{C}\setminus \R$ satisfy \eqref{5.3.1} and \eqref{5.3.2}. Then
\begin{enumerate}
    \item smooth periodic travelling wave solutions $\phi$ will exist if $r_1<r_2<r_3$ and
    \begin{enumerate}
        \item $\tilde{\gamma}>0$, with $r_2=\min\limits_{z\in\R}\phi(z)$ and $r_3=\max\limits_{z\in\R}\phi(z)$;
        \item $\tilde{\gamma}<0$, with $r_1=\min\limits_{z\in\R}\phi(z)$ and $r_2=\max\limits_{z\in\R}\phi(z)$.
    \end{enumerate}
    \item smooth travelling wave solutions $\phi$ with horizontal asymptotes will exist if
    \begin{enumerate}
        \item $r_1=r_2<r_3$ and $\tilde{\gamma}>0$ with $r_2=\inf\limits_{z\in\R}\phi(z)$, $r_3 = \max\limits_{z\in\R}\phi(z)$ and $\phi\downarrow r_2$ as $z\to\pm\infty$;
        \item $r_1<r_2=r_3$ and $\tilde{\gamma}<0$, with $r_1=\min\limits_{z\in\R}\phi(z)$, $r_2 = \sup\limits_{z\in\R}\phi(z)$ and $\phi\uparrow r_2$ as $z\to\pm\infty$.
    \end{enumerate}
\end{enumerate}
\end{theorem}

Still for $\gamma\neq 0$, assume now that $P(\phi)$ has all five real zeros $r_1\leq r_2\leq r_3\leq r_4\leq r_5$ counted without their multiplicities. In this case, we write
\begin{align*}
    \gamma \phi^5 + \beta \phi^4 - \phi^3 + (\alpha+c)\phi^2+ 2A\phi + B = \gamma(\phi-r_1)(\phi-r_2)(\phi-r_3)(\phi-r_4)(\phi-r_5).
\end{align*}
After a new comparison of coefficients, we obtain the following compatibility conditions:
\begin{align}\label{5.3.3}
    \begin{aligned}
    r_1 =& -\f{\beta + \gamma(r_2+r_3+r_4+r_5)}{\gamma},\\
    \gamma(r_2^2+&r_3^2+r_4^2+r_5^2) + \beta(r_2+r_3+r_4+r_5)+\gamma(r_2+r_3)(r_4+r_5) + \gamma(r_2r_3+r_4r_5)=1,\\
    \beta(r_2+&r_3)(r_4+r_5) +\beta(r_2r_3+r_4r_5)+\gamma(r_2+r_3)(r_2r_3+r_4r_5)\\
    &+ \gamma r_2r_3(r_2+r_3)+\gamma(r_2+r_3)(r_4^2+r_5^2+r_4r_5) +\gamma(r_4+r_5)(r_2^2+r_3^2+r_2r_3)=c+\alpha,\\
    2A =& -\beta r_2r_3(r_4+r_5) -\beta r_4r_5(r_2+r_3)-\gamma(r_2r_3+r_4r_5)(r_2+r_3)(r_4+r_5)-\gamma r_2r_3(r_4+r_5)^2\\
    &-\gamma r_4r_5(r_2^2+r_3^2+r_2r_3),\\
    B=& -\gamma r_1r_2r_3r_4r_5.
    \end{aligned}
\end{align}
In this case, we have more possibilities regarding the existence of double zeros. In summary, we must have one of the following
\begin{enumerate}
    \item all zeros are simple: $r_1<r_2<r_3<r_4<r_5$. For all the possible distributions of $\phi$ among the zeros, we have the cases:
    \begin{itemize}
        \item if $r_1\leq\phi\leq r_2<r_3<r_4<r_5$, then $-P(\phi)/\Gamma>0$ if and only if $\tilde{\gamma}<0$. In this case there exists a smooth periodic travelling wave solution $\phi$ with $r_1=\min\limits_{z\in\R}\phi(z)$ and $r_2=\max\limits_{z\in\R}\phi(z)$;
        
        \item if $r_1<r_2\leq\phi\leq r_3<r_4<r_5$, then $-P(\phi)/\Gamma>0$ if and only if $\tilde{\gamma}>0$ and there exists a smooth periodic travelling wave solution $\phi$ with $r_2=\min\limits_{z\in\R}\phi(z)$ and $r_3=\max\limits_{z\in\R}\phi(z)$;
        
        \item if $r_1<r_2<r_3\leq\phi\leq r_4<r_5$, then $-P(\phi)/\Gamma>0$ if and only if $\tilde{\gamma}<0$ and there exists a smooth periodic travelling wave solution $\phi$ with $r_3=\min\limits_{z\in\R}\phi(z)$ and $r_4=\max\limits_{z\in\R}\phi(z)$;
        
        \item if $r_1<r_2<r_3<r_4\leq \phi\leq r_5$, then  $-P(\phi)/\Gamma>0$ if and only if $\tilde{\gamma}>0$ and there exists a smooth periodic travelling wave solution $\phi$ with $r_4=\min\limits_{z\in\R}\phi(z)$ and $r_5=\max\limits_{z\in\R}\phi(z)$;
    \end{itemize}
    \item three zeros are simple and one is double:
    \begin{itemize}
        \item if $r_1=r_2<\phi\leq r_3<r_4<r_5$, then $-P(\phi)/\Gamma>0$ if and only if $\tilde{\gamma}>0$, and there exists a smooth travelling wave solution $\phi$ with $r_2=\inf\limits_{z\in\R}\phi(z)$, $r_3 = \max\limits_{z\in\R}\phi(z)$ and $\phi\downarrow r_2$ as $z\to\pm\infty$;
        
        \item if $r_1\leq \phi<r_2=r_3<r_4<r_5$, then $-P(\phi)/\Gamma>0$ if and only if $\tilde{\gamma}<0$, and there will exist a smooth travelling wave solution $\phi$ with $r_1=\min\limits_{z\in\R}\phi(z)$, $r_2 = \sup\limits_{z\in\R}\phi(z)$ and $\phi\uparrow r_2$ as $z\to\pm\infty$;
        
        \item if $r_1<r_2=r_3<\phi\leq r_4<r_5$, then $-P(\phi)/\Gamma>0$ if and only if $\tilde{\gamma}<0$, and there will exist a smooth travelling wave solution $\phi$ with $r_3=\inf\limits_{z\in\R}\phi(z)$, $r_4 = \max\limits_{z\in\R}\phi(z)$ and $\phi\downarrow r_3$ as $z\to\pm\infty$;
        
        \item if $r_1<r_2<r_3=r_4<\phi\leq r_5$, then $-P(\phi)/\Gamma>0$ if and only if $\tilde{\gamma}>0$, and there will exist a smooth travelling wave solution $\phi$ with $r_4=\inf\limits_{z\in\R}\phi(z)$, $r_5 = \max\limits_{z\in\R}\phi(z)$ and $\phi\downarrow r_4$ as $z\to\pm\infty$;
        
        \item if $r_1<r_2\leq\phi<r_3=r_4<r_5$, then $-P(\phi)/\Gamma>0$ if and only if $\tilde{\gamma}>0$, and there will exist a smooth travelling wave solution $\phi$ with $r_2=\min\limits_{z\in\R}\phi(z)$, $r_3 = \sup\limits_{z\in\R}\phi(z)$ and $\phi\uparrow r_3$ as $z\to\pm\infty$;
        
        \item if $r_1<r_2<r_3\leq \phi < r_4=r_5$, then $-P(\phi)/\Gamma>0$ if and only if $\tilde{\gamma}<0$, and there will exist a smooth travelling wave solution $\phi$ with $r_3=\min\limits_{z\in\R}\phi(z)$, $r_4 = \sup\limits_{z\in\R}\phi(z)$ and $\phi\uparrow r_4$ as $z\to\pm\infty$.
    \end{itemize}
    \item only one zero is simple and two are double:
    \begin{itemize}
        \item if $r_1=r_2<\phi<r_3=r_4<r_5$, then $-P(\phi)/\Gamma>0$ if and only if $\tilde{\gamma}>0$, and there will exist a smooth traveling wave solution $\phi$ with $r_2=\inf\limits_{z\in\R}\phi(z)$, $r_3 = \sup\limits_{z\in\R}\phi(z)$ and $\phi\downarrow r_2$, $\phi\uparrow r_3$ as $z\to\pm\infty$;
        
        \item if $r_1<r_2=r_3<\phi<r_4=r_5$, then $-P(\phi)/\Gamma>0$ if and only if $\tilde{\gamma}<0$, and there will exist a smooth traveling wave solution $\phi$ with $r_3=\inf\limits_{z\in\R}\phi(z)$, $r_4 = \sup\limits_{z\in\R}\phi(z)$ and $\phi\downarrow r_3$, $\phi\uparrow r_4$ as $z\to\pm\infty$.
    \end{itemize}
\end{enumerate}

The cases obtained above lead to the following theorem:

\begin{theorem}\label{teo5.2}\textbf{(Case $\gamma\neq 0$ with five real zeros)}
Let $\gamma\neq 0$, $\tilde{\gamma}=\gamma/\Gamma$ and suppose $r_1\leq r_2\leq r_3\leq r_4\leq r_5\in\R$ satisfy \eqref{5.3.3}. Then
\begin{enumerate}
    \item smooth periodic travelling wave solutions $\phi$ will exist if $r_1<r_2<r_3<r_4<r_5$ and
    \begin{enumerate}
        \item $\tilde{\gamma}>0$, with $r_2=\min\limits_{z\in\R}\phi(z)$ and $r_3=\max\limits_{z\in\R}\phi(z)$;
        \item $\tilde{\gamma}>0$, with $r_4=\min\limits_{z\in\R}\phi(z)$ and $r_5 = \max\limits_{z\in\R}\phi(z)$;
        \item $\tilde{\gamma}<0$, with $r_1=\min\limits_{z\in\R}\phi(z)$ and $r_2=\max\limits_{z\in\R}\phi(z)$;
        \item $\tilde{\gamma}<0$, with $r_3=\min\limits_{z\in\R}\phi(z)$ and $r_4 = \max\limits_{z\in\R}\phi(z)$.
    \end{enumerate}
    \item smooth travelling wave solutions $\phi$ with horizontal asymptotes will exist if
    \begin{enumerate}
        \item $r_1=r_2<r_3<r_4<r_5$ and $\tilde{\gamma}>0$, with $r_2=\inf\limits_{z\in\R}\phi(z)$, $r_3 = \max\limits_{z\in\R}\phi(z)$ and  $\phi\downarrow r_2$ as $z\to\pm\infty$;
        
        \item $r_1=r_2<r_3=r_4<r_5$ and $\tilde{\gamma}>0$, with $r_2=\inf\limits_{z\in\R}\phi(z)$, $r_3 = \sup\limits_{z\in\R}\phi(z)$ and $\phi\uparrow r_3$, $\phi\downarrow r_2$ as $z\to\pm\infty$.
        
        \item $r_1<r_2<r_3=r_4<r_5$ and $\tilde{\gamma}>0$, with $r_4=\inf\limits_{z\in\R}\phi(z)$, $r_5 = \max\limits_{z\in\R}\phi(z)$ and $\phi\downarrow r_4$ as $z\to\pm\infty$;
        
        \item $r_1<r_2<r_3=r_4<r_5$ and $\tilde{\gamma}>0$, with $r_2=\min\limits_{z\in\R}\phi(z)$, $r_3 = \sup\limits_{z\in\R}\phi(z)$ and $\phi\uparrow r_3$ as $z\to\pm\infty$;
        
        \item $r_1<r_2=r_3<r_4<r_5$ and $\tilde{\gamma}<0$, with $r_3=\inf\limits_{z\in\R}\phi(z)$, $r_4 = \max\limits_{z\in\R}\phi(z)$ and $\phi\downarrow r_3$ as $z\to\pm\infty$;
        
        \item $r_1<r_2=r_3<r_4<r_5$ and $\tilde{\gamma}<0$, with $r_1=\min\limits_{z\in\R}\phi(z)$, $r_2 = \sup\limits_{z\in\R}\phi(z)$ and $\phi\uparrow r_2$ as $z\to\pm\infty$;
        
        \item $r_1<r_2<r_3<r_4=r_5$ and $\tilde{\gamma}<0$, with $r_3=\min\limits_{z\in\R}\phi(z)$, $r_4 = \sup\limits_{z\in\R}\phi(z)$ and $\phi\uparrow r_4$ as $z\to\pm\infty$;
        
        \item $r_1<r_2=r_3<r_4=r_5$ and $\tilde{\gamma}>0$, with $r_3=\inf\limits_{z\in\R}\phi(z)$, $r_4 = \sup\limits_{z\in\R}\phi(z)$ and $\phi\uparrow r_4$ and $\phi\downarrow r_3$ as $z\to\pm\infty$.
    \end{enumerate}
\end{enumerate}
\end{theorem}

The classification presented so far for $\gamma\neq 0$ does not contain the case where $P(\phi)$ has only one real zero. In fact, it $P(\phi)$ has only one real zero $r$ so that $P(\phi) = \gamma(\phi-r)|\phi-z_0|^2|\phi-z_1|^2$, for certain complex numbers $z_0,z_1$, then both conditions $\phi>r$ or $\phi<r$ will lead to the non-existence of bounded solutions.

Now consider $\gamma=0$ and $\beta\neq 0$ and suppose the polynomial $P(\phi)$ has two real zeros $r_1\leq r_2$ counted without their multiplicities and one complex zero $z_0$ so we can write
\begin{align*}
    \beta \phi^4 - \phi^3 + (\alpha+c)\phi^2+ 2A\phi + B = \beta(\phi-r_1)|\phi-z_0|^2.
\end{align*}
The compatibility conditions read
\begin{align}\label{5.3.4}
    \begin{aligned}
    Re(z_0) = \f{1-\beta(r_1+r_2)}{\beta},\\
    \beta |z_0|^2 - \beta(r_1+r_2)^2 +\beta r_1r_2+r_1+r_2 = c+\alpha\\
    2A = \beta r_1r_2(r_1+r_2) - r_1r_2 - \beta (r_1+r_2)|z_0|^2,\\
    B=\beta r_1r_2|z_0|^2
    \end{aligned}
\end{align}

In this case, the only possibility will be that the zeros are simple and $r_1<\phi<r_2$. Then $-P(\phi)/\Gamma>0$ if and only if $\tilde{\beta}:= \beta/\Gamma>0$ and we have the following result:

\begin{theorem}\label{teo5.3}\textbf{(Case $\gamma=0$ with two real zeros)}
Let $\gamma=0$, $\beta\neq 0$ and suppose $r_1< r_2$ satisfy \eqref{5.3.4}. If $\tilde{\beta}=\beta/\Gamma>0$, then there exists a smooth periodic travelling wave solution $\phi$ of \eqref{1.0.4} with $r_1=\min\limits_{z\in\R}\phi(z)$ and $r_2=\max\limits_{z\in\R}\phi(z)$.
\end{theorem}

Moving on to the case where $\gamma=0$ and $P(\phi)$ has all four real zeros $r_1<r_2<r_3<r_4$, the compatibility conditions are given by
\begin{align}\label{5.3.5}
    \begin{aligned}
    r_1 = \frac{1-\beta(r_2+r_3+r_4)}{\beta},\\
    r_2+r_3+r_4 -\beta(r_2^2+r_3^2+r_4^2) - \beta(r_2r_3+r_2r_4+r_3r_4)=c+\alpha,\\
    2A = \beta(r3+r_4)r_2^2 + \beta(r_2+r_4)r_3^2+\beta(r_2+r_3)r_4^2 - r_2r_3-r_2r_4-r_3r_4 + 2\beta r_2r_3r_4,\\
    B=\beta r_1r_2r_3r_4,
    \end{aligned}
\end{align}
and the following conditions are possible:
\begin{enumerate}
    \item all four zeros $r_1<r_2<r_3<r_4$ are simple:
    \begin{itemize}
        \item if $r_1\leq\phi\leq r_2<r_3<r_4$, then $-P(\phi)/\Gamma>0$ if and only if $\tilde{\beta}>0$, and there will exist a smooth periodic travelling wave solution $\phi$ with $r_1 = \min\limits_{z\in\R}\phi(z)$ and $r_2 = \max\limits_{z\in\R}\phi(z)$;
        
        \item if $r_1<r_2\leq\phi\leq r_3<r_4$, then $-P(\phi)/\Gamma>0$ if and only if $\tilde{\beta}<0$, and there will exist a smooth periodic travelling wave solution $\phi$ with $r_2 = \min\limits_{z\in\R}\phi(z)$ and $r_3 = \max\limits_{z\in\R}\phi(z)$;
        
        \item if $r_1<r_2<r_3\leq\phi\leq r_4$, then $-P(\phi)/\Gamma>0$ if and only if $\tilde{\beta}>0$, and there will exist a smooth periodic travelling wave solution $\phi$ with $r_3 = \min\limits_{z\in\R}\phi(z)$ and $r_4 = \max\limits_{z\in\R}\phi(z)$.
    \end{itemize}
    
    \item two zeros are simple and one is double:
    \begin{itemize}
        \item if $r_1=r_2<\phi\leq r_3<r_4$, then $-P(\phi)/\Gamma>0$ if and only if $\tilde{\beta}<0$, and there will exist a smooth travelling wave solution $\phi$ with $r_2=\inf\limits_{z\in\R}\phi(z)$, $r_3=\max\limits_{z\in\R}\phi(z)$ and $\phi\downarrow r_2$ as $z\to\pm\infty$;
        
        \item if $r_1\leq \phi<r_2=r_3<r_4$, then $-P(\phi)/\Gamma>0$ if and only if $\tilde{\beta}>0$, and there will exist a smooth travelling wave solution $\phi$ with $r_1=\min\limits_{z\in\R}\phi(z)$, $r_2=\sup\limits_{z\in\R}\phi(z)$ and $\phi\uparrow r_2$ as $z\to\pm\infty$;
        
        \item if $r_1<r_2=r_3<\phi\leq r_4$, then $-P(\phi)/\Gamma>0$ if and only if $\tilde{\beta}>0$, and there will exist a smooth travelling wave solution $\phi$ with $r_3=\inf\limits_{z\in\R}\phi(z)$, $r_4=\max\limits_{z\in\R}\phi(z)$ and $\phi\downarrow r_3$ as $z\to\pm\infty$;
        
        \item if $r_1<r_2\leq \phi<r_3=r_4$, then $-P(\phi)/\Gamma>0$ if and only if $\tilde{\beta}<0$, and there will exist a smooth travelling wave solution $\phi$ with $r_2=\min\limits_{z\in\R}\phi(z)$, $r_3=\sup\limits_{z\in\R}\phi(z)$ and $\phi\uparrow r_3$ as $z\to\pm\infty$.
    \end{itemize}
    
    \item two zeros are double: if $r_1=r_2<\phi<r_3=r_4$, then $-P(\phi)/\Gamma>0$ if and only if $\tilde{\beta}<0$. In this case there will exist a smooth travelling wave solution $\phi$ with $r_2 = \inf\limits_{z\in\R}\phi(z)$, $r_3 = \sup\limits_{z\in\R}\phi(z)$ and $\phi\downarrow r_2$, $\phi\uparrow r_3$ as $z\to\pm\infty$.
\end{enumerate}

The following result comes from the discussion presented above.

\begin{theorem}\label{teo5.4}\textbf{(Case $\gamma=0$ with four real zeros)}\label{teo5.4}
Let $\gamma=0$, $\beta\neq 0$, $\tilde{\beta} = \beta/\Gamma$ and suppose $r_1\leq r_2\leq r_3\leq r_4$ satisfy \eqref{5.3.5}. Then
\begin{enumerate}
    \item whenever $\tilde{\beta}>0$
    \begin{enumerate}
        \item $r_1<r_2<r_3<r_4$, there is a smooth periodic travelling wave solution $\phi$ of \eqref{1.0.4} with $r_1=\min\limits_{z\in\R}\phi(z)$ and $r_2=\max\limits_{z\in\R}\phi(z)$;
        \item $r_1<r_2=r_3<r_4$, there is a smooth travelling wave solution $\phi$ of \eqref{1.0.4} with $r_1=\min\limits_{z\in\R}\phi(z)$, $r_2=\sup\limits_{z\in\R}\phi(z)$ and $\phi\uparrow r_2$ as $z\to\pm \infty$;
        \item[(a')] $r_1<r_2<r_3<r_4$, there is a smooth periodic travelling wave solution $\phi$ of \eqref{1.0.4} with $r_3=\min\limits_{z\in\R}\phi(z)$ and $r_4=\max\limits_{z\in\R}\phi(z)$;
        \item[(b')] $r_1<r_2=r_3<r_4$, there is a smooth travelling wave solution $\phi$ of \eqref{1.0.4} with $r_2=\inf\limits_{z\in\R}\phi(z)$, $r_4=\max\limits_{z\in\R}\phi(z)$ and $\phi\downarrow r_2$ as $z\to\pm \infty$.
    \end{enumerate}
    \item whenever $\tilde{\beta}<0$ and
    \begin{enumerate}
        \item $r_1<r_2<r_3<r_4$, there is a smooth periodic travelling wave solution $\phi$ of \eqref{1.0.4} with $r_2=\min\limits_{z\in\R}\phi(z)$ and $r_3=\max\limits_{z\in\R}\phi(z)$;
        \item $r_1=r_2<r_3<r_4$, there is a smooth travelling wave solution $\phi$ of \eqref{1.0.4} with $r_2=\inf\limits_{z\in\R}\phi(z)$, $r_3=\max\limits_{z\in\R}\phi(z)$ and $\phi\downarrow r_2$ as $z\to\pm \infty$;
        \item $r_1=r_2<r_3=r_4$, there is a smooth travelling wave solution $\phi$ of \eqref{1.0.4} with $r_1=\inf\limits_{z\in\R}\phi(z)$, with $r_2=\sup\limits_{z\in\R}\phi(z)$ and $\phi\downarrow r_2$, $\phi\uparrow r_3$ as $z\to \pm\infty$;
        \item[(b')]$r_1<r_2<r_3=r_4$, there is a smooth travelling wave solution $\phi$ of \eqref{1.0.4} with $r_2=\min\limits_{z\in\R}\phi(z)$, $r_3=\sup\limits_{z\in\R}\phi(z)$ and $\phi\uparrow r_3$ as $z\to\pm \infty$;
    \end{enumerate}
\end{enumerate}
\end{theorem}

Theorem \ref{teo5.4} is the last result regarding case $\epsilon=0$. As we have just shown, due to the absence of poles in this case, it was enough to look at the sign of $P(\phi)$ and its relation to $\tilde{\gamma}$. When $\epsilon\neq 0$, as we will discuss in the next subsection, although the analysis of $P(\phi)$ is necessary, it is not sufficient to guarantee that the entire term $P(\phi)/\epsilon^2(\tilde{c}-\phi)$ is positive. However, since our discussion on the sign of $P(\phi)$ presented in this subsection will also be necessary, most of the calculations are similar to the ones we have just carried out.

\subsection{Case $\epsilon\neq0$}

Differently from the evolutive case, the case $\epsilon \neq 0$ has a pole and this will lead to the existence of weak travelling wave solutions. However, we need to know how those weak solutions will behave when $\phi$ approaches the pole. The next lemma will be of extreme importance for weak solutions as it tells that any travelling solution $\phi$ will be smooth with the exception of points $x_0$ such that $\phi(x_0)=c$.

\begin{lemma}\label{lema5.1}
Let $\alpha\in\R$ and, for $\epsilon \neq 0$, let $\tilde{c} = \f{\epsilon^2 c - \Gamma}{\epsilon^2}$. A function $\phi\in H^1_{loc}(\R)$ is a travelling wave solution for \eqref{1.0.4} with $\gamma\to 10\gamma$ and $\beta\to 6 \beta$ if and only if the following conditions hold
\begin{enumerate}
    \item[(a)] There are disjoint open intervals $E_i$, $i\geq 1$, and a closed set $C$ such that $\R\setminus C = \bigcup\limits_{i=1}^{\infty} E_i$, $\phi\in C^{\infty}(E_i)$ and
    \begin{align*}
        \begin{cases}
        \phi(z) = \tilde{c},& z\in C,\\
        \phi(z) \neq \tilde{c},& z\in\R\setminus C.
        \end{cases}
    \end{align*}
    \item[(b)] There is a constant $A\in\R$ such that for each $i\geq 1$ there exists $b_i\in\R$ such that
    \begin{align}\label{5.4.1}
        (\phi')^2 = \frac{1}{\epsilon^2}\frac{P(\phi)}{\tilde{c} - \phi}, \quad \phi\to \tilde{c}\quad \text{at any finite endpoint of}\,\,E_i,
    \end{align}
    with $$P(\phi) = \gamma\phi^5 + \beta\phi^4-\phi^3 + (c+\alpha)\phi^2+2A\phi+B_i.$$
    \item[(c)] If the Lebesgue measure of $C$ is not zero, then for the same $A$ of item $(b)$, we have $$2A = -5\gamma \tilde{c}^4 - 4\beta \tilde{c}^3 + 3\tilde{c}^2 - 2(\alpha+c)\tilde{c}.$$
    \item[(d)] $(\phi-\tilde{c})^2\in W^{2,1}_{loc}(\R)$.
\end{enumerate}
\end{lemma}
\begin{proof}
First observe that
\begin{align*}
    \left[\epsilon^2\left(\phi-\tilde{c}\right)^2\right]'' = 2\epsilon^2(\phi')^2 + 2\epsilon^2\phi\phi'' - 2(\epsilon^2c-\Gamma)\phi'',
\end{align*}
and then equation \eqref{ast} can be written as 
\begin{align}\label{5.4.2}
    2A + \left[\epsilon^2\left(\phi-\tilde{c}\right)^2\right]'' = \epsilon^2(\phi')^2- 5\gamma\phi^4 -4 \beta\phi^3+3\phi^2 - 2(\alpha+c)\phi.
\end{align}
From \cite{len} (Lemma 1, page 404), taking $v=\phi-c$ and $p(v) =\ds{ -\f{1}{\epsilon^2}(5\gamma\phi^4 + 4\beta\phi^3-3\phi^2 + 2(\alpha+c)\phi-2A)}$ we conclude that $\phi$ is smooth with the exception of points $z_0\in\R$ where $\phi(z_0)=\tilde{c}$.

Using continuity of $\phi$, the set $C:=\phi^{-1}(\tilde{c})$ is closed and, therefore, there are disjoint open sets $E_i$, $i\geq1$, such that $\R\setminus C = \bigcup\limits_{i=1}^{\infty}E_i$, $\phi\in C^{\infty}(E_i)$ and
    \begin{align*}
        \begin{cases}
        \phi(z) = \tilde{c},& z\in C,\\
        \phi(z) \neq \tilde{c},& z\in\R\setminus C,
        \end{cases}
    \end{align*}
finishing the proof of item (a).

Fixing $i\geq 1$, consider the set $E_i$ and observe that \eqref{5.4.2} holds pointwise in $E_i$. After multiplying \eqref{5.4.2} by $\phi'$ and integrating, we obtain a constant $B_i$ such that
\begin{align}\label{5.4.3}
    (\phi')^2 = \frac{1}{\epsilon^2}\frac{\gamma\phi^5 + \beta\phi^4-\phi^3 + (c+\alpha)\phi^2+2A\phi+B_i}{\tilde{c} - \phi}
\end{align}
on $E_i$ and the proof of item (b) is complete.

For the item (d), observe that the RHS of \eqref{5.4.2} is locally integrable, which means that $(\left(\phi-\tilde{c}\right)^2)''$ is also locally integrable and $\left(\phi-\tilde{c}\right)^2 \in W^{2,1}_{loc}(\R).$

Assume that the Lebesgue measure of $C$ is not zero. Since $\phi\in H^1_{loc}(\R)$ and $(\phi-\tilde{c})^2\in W^{2,1}_{loc}(\R)$, from Lemmas 1 and 2 (page 405) of \cite{len} we have that $$\phi'=0,\quad [(\phi-\tilde{c})^2]''=0,\quad \text{a.e on}\,\,C.$$ Furthermore, since \eqref{5.4.2} holds on $\R$ and $\phi\equiv \tilde{c}$ on $C$ we have 
\begin{align*}
    2A = -5\gamma\tilde{c}^4 - 4\beta \tilde{c}^3 + 3\tilde{c}^2-2(\alpha+c)\tilde{c}\quad \text{on}\,\, C
\end{align*}
and item (c) is proven.

Conversely, suppose $\phi$ satisfies (a)-(d). Let $C$, $E_i$ ($i\geq 1)$ be as in (a) and $A$ as in (b)-(c). Differentiation of \eqref{5.4.2} shows that \eqref{5.4.3} holds on $\R\setminus C$. If the Lebesgue measure of $C$ is zero, then we have that \eqref{5.4.3} holds a.e on $\R$. Now using that $(\phi-\tilde{c})^2\in W^{2,1}_{loc}(\R)$ we conclude that $\phi$ is a solution of \eqref{1.0.4} subject to the scalings $\gamma\mapsto 10\gamma$ and $\beta\mapsto 6\beta$.

Now assume the Lebesgue measure of $C$ is not zero. Then using Lemmas 1 and 2 (page 405) of \cite{len}, we have that $\phi'=0$ and $[(\phi-\tilde{c})^2]''=0$ a.e on $C$ and, joining these conditions with item (c) we conclude that \eqref{5.4.2} holds a.e on $C$ and, therefore, $\phi$ is a solution of \eqref{1.0.4} subject to the scalings $\gamma\mapsto 10\gamma$ and $\beta\mapsto 6\beta$.
\end{proof}


Similarly to the case $\epsilon=0$, we will once again analyse the number of zeros of the polynomial $P(\phi)$ and its sign. However, since the quadrature here is given by
$$(\phi')^2 = \f{P(\phi)}{\epsilon^2(\tilde{c}-\phi)} =: F(\phi),$$
we are obligated to consider the sign of $\tilde{c}-\phi$ and its implications on the sign of $F(\phi)$. As mentioned previously, the case $(\gamma,\beta)=(0,0)$, which has its classification given in \cite{raspa}, will not be considered in this paper.
 
 The first case here considered will be $\gamma\neq 0$. Firstly, assume that $P(\phi)$ has only one real zero $r$ so it can be written as
 $$P(\phi) = \gamma(\phi-r)|\phi-z_0|^2|\phi-z_1|^2,$$
 where $z_0,z_1$ denote the complex zeros of $P(\phi)$. We then observe that if $r<\phi<\tilde{c}$ or $\tilde{c}<\phi<r$, then no matter the sign of $\gamma$, no bounded solutions will exist.
 
 Now assume $P(\phi)$ has three real zeros $r_1\leq r_2\leq r_3$ and a complex zero $z_0$ satisfying \eqref{5.3.1} so that we can write
 $$F(\phi)=\f{\gamma}{\epsilon^2}\f{(\phi-r_1)(\phi-r_2)(\phi-r_3)\vert \phi-z_0\vert^2}{\tilde{c}-\phi}.$$
The calculations for the existence of smooth solutions are quite similar to the ones presented in the evolutive case $\epsilon=0$. For example, if $r_1\leq\phi\leq r_2<r_3$, then
$$F(\phi)>0 \quad\text{if and only if}\quad \begin{cases}\gamma>0 \quad\text{and}\quad \tilde{c}>r_2,\\ \gamma<0 \quad\text{and}\quad \tilde{c}<r_1.\end{cases}$$
Therefore, a smooth periodic travelling wave solution $\phi$ will exist in this case with $r_1=\min\limits_{z\in\R}\phi(z)$ and $r_2=\max\limits_{z\in\R}\phi(z)$ if $\gamma>0$ and $\tilde{c}>r_2$ or $\gamma<0$ and $\tilde{c}<r_1$. The remaining cases of smooth solutions for the case of three real zeros are proven similarly and the proof will be omitted.

We shall consider now the existence of weak solutions. Peakon solutions will exist whenever the pole is removed and $F(\phi)>0$. Firstly, assume all three zeros are simple.
\begin{itemize}
	\item If $\tilde{c}=r_1\leq\phi\leq r_2<r_3$, then the pole of $F(\phi)$ is removed and $F(\phi)>0$ if and only if $\gamma<0$. For these choices, there will exist a periodic peakon solution $\phi$ with $r_1=\min\limits_{z\in\R} \phi(z)$ and $r_2=\max\limits_{z\in\R} \phi(z)$;
	
	\item If $r_1\leq \phi\leq r_2=\tilde{c}<r_3$, then $F(\phi)>0$ if and only if $\gamma>0$, and there will be a periodic peakon solution $\phi$ with $r_1=\min\limits_{z\in\R} \phi(z)$ and $r_2=\max\limits_{z\in\R} \phi(z)$;
	
	\item If $r_1<\tilde{c}=r_2\leq\phi\leq r_3$, then $F(\phi)>0$ if and only if $\gamma>0$, and there will exist a periodic peakon solution $\phi$ with $r_2=\min\limits_{z\in\R} \phi(z)$ and $r_3=\max\limits_{z\in\R} \phi(z)$;
	
	\item Finally, if $r_1<r_2\leq \phi\leq r_3=\tilde{c}$, then $F(\phi)>0$ if and only if $\gamma<0$, and there will exist a periodic peakon solution $\phi$ with $r_2=\min\limits_{z\in\R} \phi(z)$ and $r_3=\max\limits_{z\in\R} \phi(z)$, completing the case of periodic peaked solutions.
\end{itemize}
The case of peakon solutions with decay can be proven under the condition that one of the zeros is double:
\begin{itemize}
	\item If $\tilde{c}=r_1\leq\phi<r_2=r_3$, then $F(\phi)>0$ if and only if $\gamma<0$ and there will exist a peakon solution $\phi$ with $r_1=\min\limits_{z\in\R} \phi(z)$, $r_2=\sup\limits_{z\in\R} \phi(z)$ and $\phi\uparrow r_2$ as $z\to\pm\infty$;
	
	\item If $r_1=r_2<\phi\leq r_3=\tilde{c}$, then $F(\phi)>0$ if and only if $\gamma<0$ and there will exist a peakon solution $\phi$ with $r_1=\inf\limits_{z\in\R} \phi(z)$, $r_3=\max\limits_{z\in\R} \phi(z)$ and $\phi\downarrow r_1$ as $z\to\pm\infty$.
\end{itemize}

For cuspon solutions, we must assume that $\phi$ reaches the value $\tilde{c}$ (which is its maximum or minimum at the cusp) so that $\phi'$ can blow-up at this point and the pole is not removed:
\begin{itemize}
	\item If $r_1\leq\phi\leq \tilde{c}<r_2<r_3$, then $F(\phi)>0$ if and only if $\gamma>0$, and there will exist a periodic cusped solution $\phi$ with $r_1=\min\limits_{z\in\R} \phi(z)$ and $\tilde{c}=\max\limits_{z\in\R} \phi(z)$;
	
	\item If $r_1<\tilde{c}\leq\phi\leq r_2<r_3$, then $F(\phi)>0$ if and only if $\gamma<0$, and there will exist a periodic cuspon solution $\phi$ with $\tilde{c}=\min\limits_{z\in\R} \phi(z)$ and $r_2=\max\limits_{z\in\R} \phi(z)$;
	\item If $r_1<r_2\leq \phi\leq \tilde{c}<r_3$, then $F(\phi)>0$ if and only if $\gamma<0$, and there will exist a periodic cuspon solution $\phi$ with $r_2=\min\limits_{z\in\R} \phi(z)$ and $\tilde{c}=\max\limits_{z\in\R} \phi(z)$;
	
	\item If $r_1<r_2<\tilde{c}\leq \phi\leq r_3$, then $F(\phi)>0$ if and only if $\gamma>0$, and there will exist a periodic cuspon solution $\phi$ with $\tilde{c}=\min\limits_{z\in\R} \phi(z)$ and $r_3=\max\limits_{z\in\R} \phi(z)$.
\end{itemize} 
Similarly to peakon solutions with decay, cuspon solutions with decay will occur when one of the zeros is double:
\begin{itemize}
	\item If $r_1=r_2<\phi\leq\tilde{c}<r_3$, then $F(\phi)>0$ if and only if $\gamma<0$, and there will be a cuspon solution $\phi$ with $r_2=\inf\limits_{z\in\R} \phi(z)$, $\tilde{c}=\max\limits_{z\in\R} \phi(z)$ and $\phi\downarrow r_2$ as $z\to\pm\infty$;
	
	\item If $r_1<\tilde{c}\leq\phi<r_2=r_3$, then $F(\phi)>0$ if and only if $\gamma<0$ and there will be a cuspon solution $\phi$ with $\tilde{c}=\min\limits_{z\in\R} \phi(z)$, $r_2=\sup\limits_{z\in\R} \phi(z)$ and $\phi\uparrow r_2$ as $z\to\pm\infty$.
\end{itemize}
The next theorem summarizes the discussion presented above.

\begin{theorem}\label{teo5.5}{\textbf{(Case $\gamma\neq0$ with three real zeros)}}
Let $\gamma, \epsilon \neq 0$, $\tilde{c} = \ds{\f{\epsilon^2 c - \Gamma}{\epsilon^2}}$ and $r_1\leq r_2\leq r_3\in\R$ satisfy \eqref{5.3.1}. Then
\begin{enumerate}
    \item smooth periodic travelling wave solutions $\phi$ will exist if $r_1<r_2<r_3$,
    \begin{enumerate}
        \item $\gamma> 0$ and $\tilde{c}>r_2$, with $r_1=\min\limits_{z\in\R}\phi(z)$ and $r_2=\max\limits_{z\in\R}\phi(z)$;
        
        \item  $\gamma> 0$ and $\tilde{c}<r_2$, with $r_2=\min\limits_{z\in\R}\phi(z)$ and $r_3=\max\limits_{z\in\R}\phi(z)$;
        
        \item $\gamma<0$ and $\tilde{c}<r_1$, with $r_1=\min\limits_{z\in\R}\phi(z)$ and $r_2=\max\limits_{z\in\R}\phi(z)$;
        
        \item $\gamma<0$ and $\tilde{c}>r_3$, with $r_2=\min\limits_{z\in\R}\phi(z)$ and $r_3=\max\limits_{z\in\R}\phi(z)$.
    \end{enumerate}

    \item smooth solutions $\phi$ with horizontal asymptotes will exist if
    \begin{enumerate}
        \item $r_1=r_2<r_3$, $\gamma> 0$ and $\tilde{c}<r_2$, with $r_2=\inf\limits_{z\in\R}\phi(z)$ and $r_3=\max\limits_{z\in\R}\phi(z)$;
        
        \item $r_1<r_2=r_3$, $\gamma> 0$ and $\tilde{c} > r_2$, with $r_1=\min\limits_{z\in\R}\phi(z)$ and $r_2=\sup\limits_{z\in\R}\phi(z)$;
        
        \item $r_1=r_2<r_3$, $\gamma<0$ and $\tilde{c}>r_3$, with $r_2=\inf\limits_{z\in\R}\phi(z)$ and $r_3=\max\limits_{z\in\R}\phi(z)$;
        
        \item $r_1<r_2=r_3$, $\gamma<0$ and $\tilde{c}<r_1$, with $r_1=\min\limits_{z\in\R}\phi(z)$ and $r_2=\sup\limits_{z\in\R}\phi(z)$.
    \end{enumerate}

    \item periodic peaked travelling wave solutions $\phi$ will exist if
    \begin{enumerate}
        \item $r_1<r_2=\tilde{c}<r_3$ and $\gamma> 0$, with $r_1=\min\limits_{z\in\R}\phi(z)$ and $r_2=\max\limits_{z\in\R}\phi(z)$;
        
        \item $r_1<\tilde{c}=r_2<r_3$ and $\gamma> 0$, with $r_2=\min\limits_{z\in\R}\phi(z)$ and $r_3=\max\limits_{z\in\R}\phi(z)$;
        
        \item $\tilde{c}=r_1<r_2<r_3$ and $\gamma<0$, with $r_1=\min\limits_{z\in\R}\phi(z)$ and $r_2=\max\limits_{z\in\R}\phi(z)$;
        
        \item $r_1<r_2<r_3=\tilde{c}$ and $\gamma<0$, with $r_2=\min\limits_{z\in\R}\phi(z)$ and $r_3=\max\limits_{z\in\R}\phi(z)$.
    \end{enumerate}
    
    \item peaked travelling wave solutions $\phi$ with decay will exist if
    \begin{enumerate}
        \item $\tilde{c}=r_1<r_2=r_3$ and $\gamma<0$, with $r_1=\min\limits_{z\in\R}\phi(z)$, $r_2=\sup\limits_{z\in\R}\phi(z)$ and $\phi\uparrow r_2$ as $z\to\pm\infty$;
        
        \item $r_1=r_2<r_3=\tilde{c}$ and $\gamma<0$, with $r_2=\inf\limits_{z\in\R}\phi(z)$, $r_3=\max\limits_{z\in\R}\phi(z)$ and $\phi \downarrow r_2$ as $z\to\pm\infty$.
        \end{enumerate}
    
    \item periodic cusped travelling wave solutions $\phi$ will exist if 
    \begin{enumerate}
        \item $r_1<\tilde{c}<r_2<r_3$ and $\gamma> 0$, with $r_1=\min\limits_{z\in\R}\phi(z)$ and $\tilde{c}=\max\limits_{z\in\R}\phi(z)$;
        
        \item $r_1<r_2<\tilde{c}<r_3$ and $\gamma> 0$, with $\tilde{c}=\min\limits_{z\in\R}\phi(z)$ and $r_3=\max\limits_{z\in\R}\phi(z)$.
        
        \item $r_1<\tilde{c}<r_2<r_3$ and $\gamma<0$, with $\tilde{c}=\min\limits_{z\in\R}\phi(z)$ and $r_2=\max\limits_{z\in\R}\phi(z)$;
        
        \item $r_1<r_2<\tilde{c}<r_3$ and $\gamma<0$, with $r_2=\min\limits_{z\in\R}\phi(z)$ and $\tilde{c}=\max\limits_{z\in\R}\phi(z)$.
    \end{enumerate}
    
    \item cusped travelling wave solutions $\phi$ with decay will exist if
    \begin{enumerate}
        \item $r_1<\tilde{c}<r_2=r_3$ and $\gamma<0$, with $\tilde{c}=\min\limits_{z\in\R}\phi(z)$, $r_2=\sup\limits_{z\in\R}\phi(z)$ and $\phi\uparrow r_2$ as $z\to\pm\infty$;
        \item $r_1=r_2<\tilde{c}<r_3$ and $\gamma<0$, with $r_2=\inf\limits_{z\in\R}\phi(z)$,
        $\tilde{c}=\max\limits_{z\in\R}\phi(z)$ and $\phi\downarrow r_2$ as $z\to\pm\infty$.
    \end{enumerate}
\end{enumerate}
\end{theorem}

To finish the case where $\gamma\neq 0$, we present the existence classification for when $P(\phi)$ has all five zeros $r_1\leq r_2\leq r_3\leq r_4\leq r_5\in \R$ satisfying \eqref{5.3.3} so that $$F(\phi) = \f{\gamma}{\epsilon^2}\f{(\phi-r_2)(\phi-r_2)(\phi-r_3)(\phi-r_4)(\phi-r_5)}{\tilde{c}-\phi}.$$ We will omit the proof due to the exhaustive use of the same arguments of Theorem \ref{teo5.5} (and all other previous theorems in this section) that lead to a quite long and repetitive proof. We decide, however, to separate the classification of this case in two different but complementary results: $\gamma>0$ and $\gamma<0$, respectively.

\begin{theorem}\label{teo5.6}{\textbf{(Case $\gamma> 0$ with five real zeros)}}
Let $\gamma >0$, $\epsilon \neq 0$, $\tilde{c} = \ds{\f{\epsilon^2 c - \Gamma}{\epsilon^2}}$ and $r_1\leq r_2\leq r_3\leq r_4\leq r_5\in\R$ satisfy \eqref{5.3.3}. Then
\begin{enumerate}
    \item smooth periodic travelling wave solutions $\phi$ will exist if $r_1<r_2<r_3<r_4<r_5$ and
    \begin{enumerate}
        \item  $\tilde{c}>r_2$, with $r_1=\min\limits_{z\in\R}\phi(z)$ and $r_2=\max\limits_{z\in\R}\phi(z)$;

        \item $\tilde{c}<r_2$, with $r_2=\min\limits_{z\in\R}\phi(z)$ and $r_3=\max\limits_{z\in\R}\phi(z)$;

        \item $\tilde{c}>r_4$, with $r_3=\min\limits_{z\in\R}\phi(z)$ and $r_4=\max\limits_{z\in\R}\phi(z)$;

        \item $\tilde{c}<r_4$, with $r_4=\min\limits_{z\in\R}\phi(z)$ and $r_5=\max\limits_{z\in\R}\phi(z)$.
    \end{enumerate}
    \item smooth travelling wave solutions $\phi$ with horizontal asymptotes will exist if
    \begin{enumerate}
        \item $r_1=r_2<r_3<r_4<r_5$ and $\tilde{c}<r_2$, with $r_2=\inf\limits_{z\in\R}\phi(z)$, $r_3=\max\limits_{z\in\R}\phi(z)$ and $\phi\downarrow r_2$ as $z\to\pm\infty$;
        \item $r_1<r_2=r_3<r_4<r_5$ and $\tilde{c}>r_3$, with $r_1=\min\limits_{z\in\R}\phi(z)$, $r_2=\sup\limits_{z\in\R}\phi(z)$ and $\phi\uparrow r_2$ as $z\to\pm\infty$;
        \item $r_1<r_2=r_3<r_4<r_5$ and $\tilde{c}>r_4$, with $r_3=\inf\limits_{z\in\R}\phi(z)$, $r_4=\max\limits_{z\in\R}\phi(z)$ and $\phi\downarrow r_3$ as $z\to\pm\infty$;
        \item $r_1<r_2<r_3<r_4=r_5$ and $\tilde{c}>r_4$, with $r_3=\min\limits_{z\in\R}\phi(z)$, $r_4=\sup\limits_{z\in\R}\phi(z)$ and $\phi\uparrow r_4$ as $z\to\pm\infty$;
        \item $r_1<r_2<r_3=r_4<r_5$ and $\tilde{c}<r_2$, with $r_2=\min\limits_{z\in\R}\phi(z)$, $r_3=\sup\limits_{z\in\R}\phi(z)$ and $\phi\uparrow r_3$ as $z\to \pm\infty$;
        \item $r_1<r_2<r_3=r_4<r_5$ and $\tilde{c}<r_4$, with $r_4=\inf\limits_{z\in\R}\phi(z)$, $r_5=\max\limits_{z\in\R}\phi(z)$ and $\phi\downarrow r_4$ as $z\to\pm\infty$;
        \item $r_1=r_2<r_3=r_4<r_5$ and $\tilde{c}<r_2$, with
        $r_2=\inf\limits_{z\in\R}\phi(z)$, $r_3=\sup\limits_{z\in\R}\phi(z)$ and $\phi\downarrow r_2$, $\phi\uparrow r_3$ as $z\to\pm\infty$;
        \item $r_1<r_2=r_3<r_4=r_5$ and $\tilde{c}>r_4$, with
        $r_3=\inf\limits_{z\in\R}\phi(z)$, $r_4=\sup\limits_{z\in\R}\phi(z)$ and $\phi\downarrow r_3$, $\phi\uparrow r_4$ as $z\to\pm\infty$.
    \end{enumerate}
    \item periodic peaked travelling wave solutions $\phi$ will exist if
    \begin{enumerate}
        \item $r_1<r_2=\tilde{c}<r_3<r_4<r_5$, with $r_1=\min\limits_{z\in\R}\phi(z)$ and $r_2=\max\limits_{z\in\R}\phi(z)$;
        
        \item $r_1<\tilde{c}=r_2<r_3<r_4<r_5$, with $r_2=\min\limits_{z\in\R}\phi(z)$ and $r_3=\max\limits_{z\in\R}\phi(z)$;
        
        \item $r_1<r_2<r_3<r_4=\tilde{c}<r_5$, with $r_3=\min\limits_{z\in\R}\phi(z)$ and $r_4=\max\limits_{z\in\R}\phi(z)$;
        
        \item $r_1<r_2<r_3<\tilde{c}=r_4<r_5$, with $r_4=\min\limits_{z\in\R}\phi(z)$ and $r_5=\max\limits_{z\in\R}\phi(z)$.
    \end{enumerate}
    \item peaked travelling wave solutions $\phi$ with decay on infinity will exist if
    \begin{enumerate}
        \item $r_1<\tilde{c}=r_2<r_3=r_4<r_5$, with $r_2=\min\limits_{z\in\R}\phi(z)$, $r_3=\sup\limits_{z\in\R}\phi(z)$ and $\phi\uparrow r_3$ as $z\to\pm\infty$;
        
        \item $r_1<r_2=r_3<r_4=\tilde{c}<r_5$, with $r_3=\inf\limits_{z\in\R}\phi(z)$, $r_4=\max\limits_{z\in\R}\phi(z)$ and $\phi\downarrow r_3$ as $z\to\infty$.
    \end{enumerate}
    \item periodic cusped travelling wave solutions $\phi$ will exist if
    \begin{enumerate}
        \item $r_1<\tilde{c}<r_2<r_3<r_4<r_5$, with $r_1=\min\limits_{z\in\R}\phi(z)$ and $\tilde{c}=\max\limits_{z\in\R}\phi(z)$;
        
        \item $r_1<r_2<\tilde{c}<r_3<r_4<r_5$, with $\tilde{c}=\min\limits_{z\in\R}\phi(z)$ and $r_3=\max\limits_{z\in\R}\phi(z)$;
        
        \item $r_1<r_2<r_3<\tilde{c}<r_4<r_5$, with $r_3=\min\limits_{z\in\R}\phi(z)$ and $\tilde{c}=\max\limits_{z\in\R}\phi(z)$;
        
        \item $r_1<r_2<r_3<r_4<\tilde{c}<r_5$, with $\tilde{c}=\min\limits_{z\in\R}\phi(z)$ and $r_5=\max\limits_{z\in\R}\phi(z)$.
    \end{enumerate}
    \item cusped travelling wave solutions $\phi$ with decay on infinity will exist if
    \begin{enumerate}
        \item $r_1<r_2<\tilde{c}<r_3=r_4<r_5$, with $\tilde{c}=\min\limits_{z\in\R}\phi(z)$, $r_3=\sup\limits_{z\in\R}\phi(z)$ and $\phi\uparrow r_3$ as $z\to\pm\infty$;
        
        \item $r_1<r_2=r_3<\tilde{c}<r_4<r_5$, with $r_3=\inf\limits_{z\in\R}\phi(z)$, $\tilde{c}=\max\limits_{z\in\R}\phi(z)$ and $\phi\downarrow r_3$ as $z\to\infty$.
    \end{enumerate}
\end{enumerate}
\end{theorem}

We shall now present the version of Theorem \ref{teo5.6} for $\gamma<0$.

\begin{theorem}\label{teo5.7}{\textbf{(Case $\gamma< 0$ with five real zeros)}}\label{teo5.7}
Let $\gamma <0$, $\epsilon \neq 0$, $\tilde{c} = \ds{\f{\epsilon^2 c - \Gamma}{\epsilon^2}}$ and $r_1\leq r_2\leq r_3\leq r_4\leq r_5\in\R$ satisfy \eqref{5.3.3}. Then
\begin{enumerate}
    \item smooth periodic travelling wave solutions $\phi$ will exist if $r_1<r_2<r_3<r_4<r_5$ and
    \begin{enumerate}
        \item $\tilde{c}<r_1$, with $r_1=\min\limits_{z\in\R}\phi(z)$ and $r_2=\max\limits_{z\in\R}\phi(z)$;
        
        \item $\tilde{c}>r_3$, with $r_2=\min\limits_{z\in\R}\phi(z)$ and $r_3=\max\limits_{z\in\R}\phi(z)$;
        
        \item $\tilde{c}<r_3$, with $r_3=\min\limits_{z\in\R}\phi(z)$ and $r_4=\max\limits_{z\in\R}\phi(z)$;
        
        \item $\tilde{c}>r_5$, with $r_4=\min\limits_{z\in\R}\phi(z)$ and $r_5=\max\limits_{z\in\R}\phi(z)$.
    \end{enumerate}
    \item smooth travelling wave solutions $\phi$ with horizontal asymptotes will exist if
    \begin{enumerate}
        \item $r_1=r_2<r_3<r_4<r_5$ and $\tilde{c}>r_3$, with $r_2=\inf\limits_{z\in\R}\phi(z)$, $r_3=\max\limits_{z\in\R}\phi(z)$ and $\phi\downarrow r_2$ as $z\to\pm\infty$;
        
        \item $r_1<r_2=r_3<r_4<r_5$ and $\tilde{c} <r_1$, with $r_1=\min\limits_{z\in\R}\phi(z)$, $r_2=\sup\limits_{z\in\R}\phi(z)$ and $\phi\uparrow r_2$ as $z\to\pm\infty$;
        
        \item $r_1<r_2=r_3<r_4<r_5$ and $\tilde{c} <r_3$, with $r_3=\inf\limits_{z\in\R}\phi(z)$, $r_4=\max\limits_{z\in\R}\phi(z)$ and $\phi\downarrow r_3$ as $z\to \pm\infty$;
        
        \item $r_1<r_2<r_3=r_4<r_5$ and $\tilde{c}>r_3$, with $r_2=\min\limits_{z\in\R}\phi(z)$, $r_3=\sup\limits_{z\in\R}\phi(z)$ and $\phi\uparrow r_3$ as $z\to\pm\infty$;
        
        \item $r_1<r_2<r_3=r_4<r_5$ and $\tilde{c} > r_5$, with $r_4=\inf\limits_{z\in\R}\phi(z)$, $r_5=\max\limits_{z\in\R}\phi(z)$ and $\phi\downarrow r_4$ as $z\to\pm\infty$;
        
        \item $r_1<r_2<r_3<r_4=r_5$ and $\tilde{c}<r_3$, with $r_3=\min\limits_{z\in\R}\phi(z)$, $r_4=\sup\limits_{z\in\R}\phi(z)$ and $\phi\uparrow r_4$ as $z\to\pm\infty$;
        
        \item $r_1=r_2<r_3=r_4<r_5$ and $\tilde{c}>r_3$, with
        $r_2=\inf\limits_{z\in\R}\phi(z)$, $r_3=\sup\limits_{z\in\R}\phi(z)$ and $\phi\downarrow r_2$, $\phi\uparrow r_3$ as $z\to\pm\infty$;
        
        \item $r_1<r_2=r_3<r_4=r_5$ and $\tilde{c}<r_3$, with
        $r_3=\inf\limits_{z\in\R}\phi(z)$, $r_4=\sup\limits_{z\in\R}\phi(z)$ and $\phi\downarrow r_3$, $\phi\uparrow r_4$ as $z\to\pm\infty$.
    \end{enumerate}
    \item periodic peaked travelling wave solutions $\phi$ will exist if
    \begin{enumerate}
        \item $\tilde{c}=r_1<r_2<r_3<r_4<r_5$, with $r_1=\min\limits_{z\in\R}\phi(z)$ and $r_2=\max\limits_{z\in\R}\phi(z)$;
        \item $r_1<r_2<r_3=\tilde{c}<r_4<r_5$, with $r_2=\min\limits_{z\in\R}\phi(z)$ and $r_3=\max\limits_{z\in\R}\phi(z)$;
        \item $r_1<r_2<\tilde{c}=r_3< r_4<r_5$, with $r_3=\min\limits_{z\in\R}\phi(z)$ and $r_4=\max\limits_{z\in\R}\phi(z)$;
        \item $r_1<r_2<r_3<r_4<r_5=\tilde{c}$, with $r_4=\min\limits_{z\in\R}\phi(z)$ and $r_5=\max\limits_{z\in\R}\phi(z)$.
    \end{enumerate}
    \item peaked travelling wave solutions $\phi$ with decay will exist if
    \begin{enumerate}
        \item $\tilde{c}=r_1<r_2=r_3<r_4<r_5$, with $r_1=\min\limits_{z\in\R}\phi(z)$, $r_2=\sup\limits_{z\in\R}\phi(z)$ and $\phi\uparrow r_2$ as $z\to\pm\infty$;
        \item $r_1=r_2<r_3=\tilde{c}<r_4<r_5$, with $r_2=\inf\limits_{z\in\R}\phi(z)$, $r_3=\max\limits_{z\in\R}\phi(z)$ and $\phi\downarrow r_2$ as $z\to\pm\infty$;
        \item $r_1<r_2<\tilde{c}=r_3<r_4=r_5$, with $r_3=\min\limits_{z\in\R}\phi(z)$, $r_4=\sup\limits_{z\in\R}\phi(z)$ and $\phi\uparrow r_4$ as $z\to\pm\infty$;
        \item $r_1<r_2<r_3=r_4<r_5=\tilde{c}$, with $r_4=\inf\limits_{z\in\R}\phi(z)$, $r_5=\max\limits_{z\in\R}\phi(z)$ and $\phi\downarrow r_4$ as $z\to\pm\infty$.
    \end{enumerate}
    \item periodic cusped travelling wave solutions $\phi$ will exist if
    \begin{enumerate}
        \item $r_1<\tilde{c}<r_2<r_3<r_4<r_5$, with $\tilde{c}=\min\limits_{z\in\R}\phi(z)$ and $r_2=\max\limits_{z\in\R}\phi(z)$;
        \item $r_1<r_2<\tilde{c}<r_3<r_4<r_5$, with $r_2=\min\limits_{z\in\R}\phi(z)$ and $\tilde{c}=\max\limits_{z\in\R}\phi(z)$;
        \item $r_1<r_2<r_3<\tilde{c}<r_4<r_5$, wtih $\tilde{c}=\min\limits_{z\in\R}\phi(z)$ and $r_4=\max\limits_{z\in\R}\phi(z)$;
        \item $r_1<r_2<r_3<r_4<\tilde{c}<r_5$, with $r_4=\min\limits_{z\in\R}\phi(z)$ and $\tilde{c}=\max\limits_{z\in\R}\phi(z)$.
    \end{enumerate}

    \item cusped travelling wave solutions $\phi$ with decay will exist if
    \begin{enumerate}
        \item $r_1<\tilde{c}<r_2=r_3<r_4<r_5$, with $\tilde{c}=\min\limits_{z\in\R}\phi(z)$, $r_2=\sup\limits_{z\in\R}\phi(z)$ and $\phi\uparrow r_2$ as $z\to\pm\infty$;

        \item $r_1=r_2<\tilde{c}<r_4<r_5$, with $r_1=\inf\limits_{z\in\R}\phi(z)$, $\tilde{c}=\max\limits_{z\in\R}\phi(z)$ and $\phi\downarrow r_1$ as $z\to\pm\infty$;

        \item $r_1<r_2<r_3<\tilde{c}<r_4=r_5$, with $\tilde{c}=\min\limits_{z\in\R}\phi(z)$, $r_4=\sup\limits_{z\in\R}\phi(z)$ and $\phi\uparrow r_4$ as $z\to\pm\infty$;

        \item $r_1<r_2<r_3=r_4<\tilde{c}<r_5$, with $r_4=\inf\limits_{z\in\R}\phi(z)$, $\tilde{c}=\max\limits_{z\in\R}\phi(z)$ and $\phi\downarrow r_4$ as $z\to\pm\infty$.
    \end{enumerate}
\end{enumerate}
\end{theorem}

With the last two theorems we finish the classification for $\gamma\neq 0$. We now proceed with the case where $\gamma=0$ and the quartic term $u^3u_x$ is eliminated from the equation. The proofs for this case will also be omitted due to their lenght and repetition of the arguments of previous theorems.

\begin{theorem}\label{teo5.8}{\textbf{(Case $\gamma=0$, $\beta\neq0$ with two real zeros)}}\label{teo5.9}
Let $\gamma =0$, $\epsilon \neq 0$, $\beta\neq0$, $\tilde{c} = \ds{\f{\epsilon^2 c - \Gamma}{\epsilon^2}}$ and $r_1\leq r_2\in\R$ satisfy \eqref{5.3.4}. Then
\begin{enumerate}
    \item smooth periodic travelling wave solutions $\phi$ will exist if $r_1<r_2$ and
    \begin{enumerate}
        \item $\tilde{c}<r_1$ and $\beta>0$, with $r_1=\min\limits_{z\in\R}\phi(z)$ and $r_2=\max\limits_{z\in\R}\phi(z)$;
        \item $\tilde{c}>r_2$ and $\beta<0$, with $r_1=\min\limits_{z\in\R}\phi(z)$ and $r_2=\max\limits_{z\in\R}\phi(z)$.
    \end{enumerate}
    \item periodic peaked travelling wave solutions $\phi$ will exist if
    \begin{enumerate}
        \item $\tilde{c}=r_1<r_2$ and $\beta>0$, with $r_1=\min\limits_{z\in\R}\phi(z)$ and $r_2=\max\limits_{z\in\R}\phi(z)$;
        \item $r_1<r_2=\tilde{c}$ and $\beta<0$, with $r_1=\min\limits_{z\in\R}\phi(z)$ and $r_2=\max\limits_{z\in\R}\phi(z)$.
    \end{enumerate}
    \item periodic cusped solutions $\phi$ will exist if $r_1<\tilde{c}<r_2$ and
    \begin{enumerate}
        \item $\beta>0$, with $\tilde{c}=\min\limits_{z\in\R}\phi(z)$ and $r_2=\max\limits_{z\in\R}\phi(z)$;
        \item $\beta<0$, with $r_1=\min\limits_{z\in\R}\phi(z)$ and
        $\tilde{c}=\max\limits_{z\in\R}\phi(z)$.
    \end{enumerate}
\end{enumerate}
\end{theorem}

Next we consider the case of four real zeros and also separate the cases $\beta>0$ and $\beta<0$:

\begin{theorem}\label{teo5.9}{\textbf{(Case $\gamma=0$, $\beta>0$ with four real zeros)}}
Let $\gamma =0$, $\epsilon \neq 0$, $\beta>0$, $\tilde{c} = \ds{\f{\epsilon^2 c - \Gamma}{\epsilon^2}}$ and $r_1\leq r_2\leq r_3\leq r_4\in\R$ satisfy \eqref{5.3.5}. Then
\begin{enumerate}
    \item smooth periodic travelling wave solutions $\phi$ will exist if $r_1<r_2<r_3<r_4$ and
    \begin{enumerate}
        \item $\tilde{c}<r_1$, with $r_1=\min\limits_{z\in\R}\phi(z)$ and $r_2=\max\limits_{z\in\R}\phi(z)$;
        \item $\tilde{c}>r_3$, with $r_2=\min\limits_{z\in\R}\phi(z)$ and $r_3=\max\limits_{z\in\R}\phi(z)$.
        \item $\tilde{c}<r_3$, with $r_3=\min\limits_{z\in\R}\phi(z)$ and $r_4=\max\limits_{z\in\R}\phi(z)$.
    \end{enumerate}
    \item smooth travelling wave solutions $\phi$ with horizontal asymptotes will exist if
    \begin{enumerate}
        \item $r_1=r_2<r_3<r_4$ and $\tilde{c}>r_3$, with $r_2=\inf\limits_{z\in\R}\phi(z)$, $r_3=\max\limits_{z\in\R}\phi(z)$ and $\phi\downarrow r_2$ as $z\to\pm\infty$;
        \item $r_1<r_2=r_3<r_4$ and $\tilde{c}<r_1$, with $r_1=\min\limits_{z\in\R}\phi(z)$,
        $r_2=\sup\limits_{z\in\R}\phi(z)$ and $\phi\uparrow r_2$ as $z\to\pm\infty$;
        \item $r_1<r_2=r_3<r_4$ and $\tilde{c}<r_3$, with $r_3=\inf\limits_{z\in\R}\phi(z)$,
        $r_4=\max\limits_{z\in\R}\phi(z)$ and $\phi\downarrow r_3$ as $z\to\pm\infty$;
        \item $r_1<r_2<r_3=r_4$ and $\tilde{c}>r_3$, with $r_2=\min\limits_{z\in\R}\phi(z)$,
        $r_3=\sup\limits_{z\in\R}\phi(z)$ and $\phi\uparrow r_3$ as $z\to\pm\infty$;
        \item $r_1=r_2<r_3=r_4$ and $\tilde{c}>r_3$, with $r_2=\inf\limits_{z\in\R} \phi(z)$, $r_3=\sup\limits_{z\in\R} \phi(z)$ and $\phi\downarrow r_2$, $\phi\uparrow r_3$ as $z\to\pm\infty$.
    \end{enumerate}
    \item periodic peaked travelling wave solutions $\phi$ will exist if
    \begin{enumerate}
        \item $\tilde{c}=r_1<r_2<r_3<r_4$, with $r_1=\min\limits_{z\in\R}\phi(z)$ and $r_2=\max\limits_{z\in\R}\phi(z)$;
        \item $r_1<r_2<r_3=\tilde{c}<r_4$, with $r_2=\min\limits_{z\in\R}\phi(z)$ and $r_3=\max\limits_{z\in\R}\phi(z)$;
        \item $r_1<r_2<r_3=\tilde{c}<r_4$, with $r_3=\min\limits_{z\in\R}\phi(z)$ and $r_4=\max\limits_{z\in\R}\phi(z)$.
    \end{enumerate}
    
    \item peaked travelling wave solutions $\phi$ with decay will exist if
    \begin{enumerate}
        \item $\tilde{c} = r_1<r_2=r_3<r_4$, with $r_1 = \min\limits_{z\in\R}\phi(z)$, $r_2=\sup\limits_{z\in\R} \phi(z)$ and $\phi\uparrow r_2$ as $z\to\pm\infty$;
        
        \item $r_1=r_2<r_3=\tilde{c}<r_4$, with $r_2 = \inf\limits_{z\in\R}\phi(z)$, $r_3=\max\limits_{z\in\R} \phi(z)$ and $\phi\downarrow r_2$ as $z\to\pm\infty$;
    \end{enumerate}
    
    \item periodic cusped travelling wave solutions $\phi$ will exist if
    \begin{enumerate}
        \item $r_1<\tilde{c}<r_2<r_3<r_4$, with $\tilde{c}=\min\limits_{z\in\R}\phi(z)$ and $r_2=\max\limits_{z\in\R}\phi(z)$;
        \item $r_1<r_2<\tilde{c}<r_3<r_4$, with $r_2=\min\limits_{z\in\R}\phi(z)$ and
        $\tilde{c}=\max\limits_{z\in\R}\phi(z)$;
        \item $r_1<r_2<r_3<\tilde{c}<r_4$, with $\tilde{c}=\min\limits_{z\in\R}\phi(z)$ and
        $r_4=\max\limits_{z\in\R}\phi(z)$;
    \end{enumerate}
    
    \item cusped travelling wave solutions $\phi$ with decay will exist if
    \begin{enumerate}
        \item $r_1<\tilde{c}<r_2=r_3<r_4$, with $\tilde{c}=\min\limits_{z\in\R}\phi(z)$, $r_2=\sup\limits_{z\in\R}\phi(z)$ and $\phi\uparrow r_2$ as $z\to\pm\infty$;
        \item $r_1=r_2<\tilde{c}<r_3<r_4$, with $r_2=\inf\limits_{z\in\R}\phi(z)$,
        $\tilde{c}=\max\limits_{z\in\R}\phi(z)$ and $\phi\downarrow r_2$ as $z\to\pm\infty$.
    \end{enumerate}
\end{enumerate}
\end{theorem}

With the next theorem we finish the results concerning the classification of bounded travelling wave solutions.

\begin{theorem}\label{teo5.10}{\textbf{(Case $\gamma=0$, $\beta<0$ with four real zeros)}}\label{teo4.11}
Let $\gamma =0$, $\epsilon \neq 0$, $\beta>0$, $\tilde{c} = \ds{\f{\epsilon^2 c - \Gamma}{\epsilon^2}}$ and $r_1\leq r_2\leq r_3\leq r_4\in\R$ satisfy \eqref{5.3.5}. Then
\begin{enumerate}
    \item smooth periodic travelling wave solutions $\phi$ will exist if $r_1<r_2<r_3<r_4$ and
    \begin{enumerate}
        \item $\tilde{c}>r_2$, with $r_1=\min\limits_{z\in\R}\phi(z)$ and $r_2=\max\limits_{z\in\R}\phi(z)$;
        
        \item $\tilde{c}<r_2$, with $r_2=\min\limits_{z\in\R}\phi(z)$ and $r_3=\max\limits_{z\in\R}\phi(z)$;
        
        \item $\tilde{c}>r_4$, with $r_3=\min\limits_{z\in\R}\phi(z)$ and $r_4=\max\limits_{z\in\R}\phi(z)$.
    \end{enumerate}
    
    \item smooth travelling wave solutions $\phi$ with horizontal asymptotes will exist if
    \begin{enumerate}
        \item $r_1=r_2<r_3<r_4$ and $\tilde{c}<r_2$, with $r_2=\inf\limits_{z\in\R}\phi(z)$, $r_3=\max\limits_{z\in\R}\phi(z)$ and $\phi\downarrow r_2$ as $z\to\pm\infty$;
        
        \item $r_1<r_2=r_3<r_4$ and $\tilde{c}>r_2$, with $r_1=\min\limits_{z\in\R}\phi(z)$,
        $r_2=\sup\limits_{z\in\R}\phi(z)$ and $\phi\uparrow r_2$ as $z\to\pm\infty$;
        
        \item $r_1<r_2=r_3<r_4$ and $\tilde{c}>r_4$, with $r_3=\inf\limits_{z\in\R}\phi(z)$,
        $r_4=\max\limits_{z\in\R}\phi(z)$ and $\phi\downarrow r_3$ as $z\to\pm\infty$;
        
        \item $r_1<r_2<r_3=r_4$ and $\tilde{c}<r_2$, with $r_2=\min\limits_{z\in\R}\phi(z)$,
        $r_3=\sup\limits_{z\in\R}\phi(z)$ and $\phi\uparrow r_3$ as $z\to\pm\infty$;
        
        \item $r_1=r_2<r_3=r_4$ and $\tilde{c}<r_2$, with $r_2=\inf\limits_{z\in\R} \phi(z)$, $r_3=\sup\limits_{z\in\R} \phi(z)$ and $\phi\downarrow r_2$, $\phi\uparrow r_3$ as $z\to\pm\infty$.
    \end{enumerate}
    
    \item periodic peaked travelling wave solutions $\phi$ will exist if
    \begin{enumerate}
        \item $r_1<r_2=\tilde{c}<r_3<r_4$, with $r_1=\min\limits_{z\in\R}\phi(z)$ and $r_2=\max\limits_{z\in\R}\phi(z)$;
        
        \item $r_1<\tilde{c}=r_2<r_3<r_4$, with $r_2=\min\limits_{z\in\R}\phi(z)$ and $r_3=\max\limits_{z\in\R}\phi(z)$;
        
        \item $r_1<r_2<r_3<r_4=\tilde{c}$, with $r_3=\min\limits_{z\in\R}\phi(z)$ and $r_4=\max\limits_{z\in\R}\phi(z)$.
    \end{enumerate}
    
    \item peaked travelling wave solutions $\phi$ with decay will exist if
    \begin{enumerate}
        \item $r_1<\tilde{c}=r_2<r_3=r_4$, with $r_2= \min\limits_{z\in\R}\phi(z)$, $r_3=\sup\limits_{z\in\R} \phi(z)$ and $\phi\uparrow r_2$ as $z\to\pm\infty$;
        
        \item $r_1<r_2=r_3<r_4=\tilde{c}$, with $r_3 = \inf\limits_{z\in\R}\phi(z)$, $r_4=\max\limits_{z\in\R} \phi(z)$ and $\phi\downarrow r_2$ as $z\to\pm\infty$.
    \end{enumerate}
    
    \item periodic cusped solutions $\phi$ will exist if
    \begin{enumerate}
        \item $r_1<\tilde{c}<r_2<r_3<r_4$, with $r_1=\min\limits_{z\in\R}\phi(z)$ and $\tilde{c}=\max\limits_{z\in\R}\phi(z)$;
        
        \item $r_1<r_2<\tilde{c}<r_3<r_4$, with $\tilde{c}=\min\limits_{z\in\R}\phi(z)$ and
        $r_3=\max\limits_{z\in\R}\phi(z)$;
        
        \item $r_1<r_2<r_3<\tilde{c}<r_4$, with $r_3=\min\limits_{z\in\R}\phi(z)$ and
        $\tilde{c}=\max\limits_{z\in\R}\phi(z)$.
    \end{enumerate}
    
    \item cusped travelling wave solutions $\phi$ with decay will exist if
    \begin{enumerate}
        \item $r_1<r_2<\tilde{c}<r_3=r_4$, with $\tilde{c}=\min\limits_{z\in\R}\phi(z)$, $r_3=\sup\limits_{z\in\R}\phi(z)$ and $\phi\uparrow r_3$ as $z\to\pm\infty$;
        
        \item $r_1<r_2=r_3<\tilde{c}<r_4$, with $r_3=\inf\limits_{z\in\R}\phi(z)$,
        $\tilde{c}=\max\limits_{z\in\R}\phi(z)$ and $\phi\downarrow r_2$ as $z\to\pm\infty$.
    \end{enumerate}
\end{enumerate}
\end{theorem}

To finish this section, we would like to discuss the possibility of composing weak solutions to obtain new waves called \textit{composite waves}. In \cite{len}, the author shows that provided that the constant $A$ in the quadrature form for the Camassa-Holm equation (see equation \eqref{1.0.3} with $\gamma=\beta=0$) satisfies a certain condition (see item (c) of Lemma \ref{lema5.1} with the same choices for $\gamma$ and $\beta$) and the Lebesgue measure of the set $C=\phi^{-1}(\tilde{c})$ is zero, then one can glue peaked and cusped solutions to obtain composite wave solutions for the Camassa-Holm equation. If the Lebesgue measure of $C$ is not zero, then one can glue cusped solutions to obtain the \textit{stumpon} solutions provided that the coefficients are in a certain ellipsoid. In our case we strongly believe that this will also happen, but so far we have not been able to explicitly obtain the manifold from the geometrical conditions. For this reason we leave the results of composite wave solutions for an upcoming paper.

\section{Explicit travelling waves solutions}\label{sec6}

In this section we look for explicit travelling waves of equation \eqref{1.0.4}. Our main tool for constructing such types of solutions are the conservation laws we established in Section \ref{sec3} and their consequences.

\subsection{Travelling waves I: elliptic integrals}

Let $P(\phi)$ be given by (\ref{5.1.4}) and $c$ be such that $P(\phi_c)=0$, where $\phi_c:=c-\Gamma/\epsilon^2$. We can rewrite $P(\phi)$ as $P(\phi)=-\epsilon^2(\phi-\phi_c)q(\phi)$ and by $\deg{(q)}$ we denote the degree of $q$. We observe that $2\leq\deg{(q)}\leq 4$. 

\begin{theorem}\label{teo6.1}
Let $u=\phi(x-ct)$, $c\neq0$, be a travelling wave solution of $(\ref{1.0.4})$, $P(\phi)$ be the polynomial defined in $(\ref{5.1.4})$ and $q(\phi)$ such that $P(\phi)=-\epsilon^2(\phi-\phi_c)q(\phi)$, where  $\phi_c = c-\Gamma/\epsilon^2$.
\begin{enumerate}
    \item If $(\be,\gamma)\neq(0,0)$, then $3\leq \deg{(q)}\leq 4$ and we have a solution given in terms of elliptic integral
    \bb\label{6.1.1}
    \int\f{d\phi}{\sqrt{q(\phi)}}=\pm z+\lambda,
    \ee
where $\lambda$ is a constant. 

\item If $\be=\gamma=0$, then $q(\phi)=\phi^2/\epsilon^2+b\phi+d$, where $b=(\phi_c-(c+\al))/\epsilon^2$, $\epsilon^2\phi_cd=B$, and
\bb\label{6.1.2}
\int\f{d\phi}{\sqrt{\f{\phi^2}{\epsilon^2}+b\phi+d}}=\pm z+\lambda.
\ee
\end{enumerate}
\end{theorem}
\begin{proof}
The proof is straightforward and is omitted.
\end{proof}

\begin{remark}\label{rem6.1}
If $\beta=\gamma=0$, then we have the following compatibility conditions in the case 2 above: $\phi_c-b\epsilon^2=c+\alpha,\,\,(b-d)\epsilon^2=2A$ and $\epsilon^2\phi_c d=B$. We note that both $A$ and $B$ are constants of integration, see the comments between equations \eqref{5.1.3} and \eqref{5.1.4}. Then we have some freedom to choose them and this freedom is inherited by $b$ and $d$.
\end{remark}

\begin{remark}\label{rem6.2}
Equation $(\ref{6.1.2})$ can be integrated, providing
$$\epsilon\ln{\left(\epsilon^2b+2\phi+2\epsilon\sqrt{\f{\phi^2}{\epsilon^2}+b\phi+d}\right)}=\pm z+\tilde{\lambda},$$
where $\tilde{\lambda}$ is a constant of integration. Solving the last equation for $\phi$, we obtain two families of equations parametrized by a new constant $\lambda\neq0$
$$
\phi_{\pm,\lambda}(z)=\left(\f{b^2\epsilon^2-4d}{4\lambda}\right)\epsilon^2e^{\pm\f{z}{\epsilon}}+\f{\lambda}{4}e^{\f{\mp z}{\epsilon}}-\f{\epsilon^2b}{2}.
$$

Assume that 
$$
\left(\f{b^2\epsilon^2-4d}{4\lambda}\right)\epsilon^2=\f{\lambda}{4}.
$$
This means that $d,\,\Gamma$ and $\lambda$ describes the {\it hyperbolic paraboloid}
$$
d=\left(\f{\Gamma+\al\epsilon^2}{2\epsilon^3}\right)^2-\left(\f{\lambda}{2\epsilon}\right)^2
$$
with center $(\Gamma,\lambda,d)=(-\al\epsilon^2,0,0)$, where we used the fact that $b=-(\Gamma+\alpha\epsilon^2)/\epsilon^4$.
In this case, the two families of solutions degenerate in a single {\it even function}, given by
$$\phi(z)=\f{\lambda}{2}\cosh{\left(\f{z}{\epsilon}\right)}+\f{\Gamma+\alpha\epsilon^2}{2\epsilon^2}.$$

On the other hand, assuming that
$$
\left(\f{b^2\epsilon^2-4d}{4\lambda}\right)\epsilon^2=-\f{\lambda}{4},
$$
then $d,\,\Gamma$ and $\lambda$ describes the {\it elliptic paraboloid}
$$
d=\left(\f{\Gamma+\al\epsilon^2}{2\epsilon^3}\right)^2+\left(\f{\lambda}{2\epsilon}\right)^2,
$$
with center $(\Gamma,\lambda,d)=(-\al\epsilon^2,0,0)$ and we have two families of corresponding solutions given by
$$\phi_\pm(z)=\pm\f{\lambda}{2}\sinh{\left(\f{z}{\epsilon}\right)}+\f{\Gamma+\alpha\epsilon^2}{2\epsilon^2}.$$
\end{remark}

Let us analyze the results we have just obtained. To this end, we focus on the solution
$$\phi(z)=\f{\lambda}{2}\cosh{\left(\f{z}{\epsilon}\right)}+\f{\Gamma+\alpha\epsilon^2}{2\epsilon^2}$$
and the two-parameter family of paraboloids
$${\cal P}_{\epsilon,\alpha}=\left\{(\Gamma,\lambda,d)\in\R^3\,\,\text{such that}\,\,d=\left(\f{\Gamma+\al\epsilon^2}{2\epsilon^3}\right)^2-\left(\f{\lambda}{2\epsilon}\right)^2
 \right\}$$

We can embed each member of these paraboloids into a two-parameter family of four-dimensional manifolds
$${\cal M}_{\epsilon,\alpha}=\left\{(x,t,\Gamma,\lambda,d)\in\R^5\,\,\text{such that}\,\,d=\left(\f{\Gamma+\al\epsilon^2}{2\epsilon^3}\right)^2-\left(\f{\lambda}{2\epsilon}\right)^2
 \right\}.$$
 
To each point of ${\cal M}_{\epsilon,\alpha}$ we can associate a solution of \eqref{1.0.3}, given by
$$u(x,t)=\f{\lambda}{2}\cosh{\left(\f{x-ct}{\epsilon}\right)}+\f{\Gamma+\alpha\epsilon^2}{2\epsilon^2}.$$

In particular, each point $p=(\Gamma,\lambda,d)\in{\cal P}_{\epsilon,\alpha}$ gives a solution $u(x,t)$ as above and while $p$ varies on ${\cal P}_{\epsilon,\alpha}$ we have a family of solutions of \eqref{1.0.3}, with $\be=\gamma=0$, parameterized by the paraboloid ${\cal P}_{\epsilon,\alpha}$.

Finally, let us now consider travelling wave solutions using the third conserved current provided by Theorem \ref{teo4.2}. If $u(x,t)=\phi(z)$, $z=x-ct$, for some $c\neq0$, equation $D_t\sqrt{m}+D_x((u-\al)\sqrt{m})=0$ reads
$$
\f{d}{dz}\left((\phi-\al-c)\sqrt{M}\right)=0, \quad M = \phi-\epsilon^2\phi'',
$$
which yields
\bb\label{6.1.3}
(\phi-\al-c)^2(\phi-\epsilon^2\phi'')=\f{k_1}{2},
\ee
where $k_1$ is a constant of integration. We note that
$$
\phi-\epsilon^2\phi''=\f{k_1}{2(\phi-\al-c)^2}\Longrightarrow \phi\phi'-\epsilon^2\phi'\phi''=\f{k_1\phi'}{2(\phi-\al-c)^2},
$$
which, after integration, reads
\bb\label{6.1.4}
(\epsilon \phi')^2=\phi^2+\f{k_1}{\phi-\al-c}.
\ee
Let $w:=1/(\phi-\al-c)$. From equation (\ref{6.1.4}) we obtain
$$
\epsilon\f{dw}{dz}=\pm w\sqrt{k_1w^3+((\al+c)w+1)^2},
$$
which gives the solution in terms of the elliptic integral
$$
\int\f{dw}{w\sqrt{k_1w^3+((\al+c)w+1)^2}}=\pm\f{z}{\epsilon}+k_2,
$$
where $k_2$ is another constant of integration.

\subsection{Travelling waves II: peakons with exponential shape}

In this subsection we will use the third conservation law resulting from Theorem \ref{teo4.2} to guide us to construct an explicit travelling wave solution to (\ref{1.0.4}). As naturally expected due to the restriction of the respective conserved current, the solution is firstly obtained assuming that some parameters in the equation are $0$. After, we prove that such a type of solution, with exponential shape, can only be obtained with those restrictions.

We begin with equation (\ref{6.1.3}), which is a consequence of the conservation law already mentioned. If we suppose $\phi,\,\phi',\,\phi''\rightarrow0$ as $|z|\rightarrow\infty$, we then conclude that $k_1=0$. This implies that either $\phi=\al+c$, which means that $\phi$ is an arbitrary constant (in view of the arbitrariness of the constant $c$), or $\phi(z)=Ae^{z}+Be^{-z}$. None of these solutions satisfy the property of vanishing at infinity, unless $A=B=\al+c=0$.

Let us then consider solutions of (\ref{6.1.3}) with $k_1=0$ in the weak sense. We begin with the following auxiliary equation
\bb\label{6.2.1}
\psi(z)-\epsilon^2\psi''(z)=\lambda\delta(z),
\ee
where $\delta(z)$ is the well known Dirac delta function and $\lambda\neq0$ is a constant. Applying the Fourier transform ${\cal F}$ to both sides of (\ref{6.2.1}) we obtain
$$
{\cal F}(\psi)=\f{\lambda}{\sqrt{2\pi}}\f{1}{1+\epsilon^2\xi^2} \Rightarrow \psi(z)=\f{\lambda}{2\epsilon}e^{-\f{|z|}{\epsilon}}.
$$

Based on the latter fact, let $\phi(z):=\lambda e^{-|z|/\epsilon}$. Then $\phi(z)-\epsilon^2\phi''(z)=2\lambda\delta(z)$ and equation (\ref{6.1.3}), with $k_1=0$, is compatible provided that $\lambda=\al+c$. Therefore, if $\be=\gamma=0$ and $\Gamma=-\al\epsilon^2$, equation (\ref{1.0.4}) admits the travelling wave solution $u(x,t)=\phi(x-ct)$, with $\phi(z)=(\al+c)e^{-|z|/\epsilon}$. This is, actually a particular case of the following stronger result:
\begin{theorem}\label{teo6.2}
The function $u:\R^2\rightarrow\R$, defined by $u(x,t)=(\al+c)e^{-|x-ct|}$, with $\al\neq c$, is a weak solution to the equation $(\ref{1.0.4})$ if and only if $\be=\gamma=0$ and $\Gamma=-\al\epsilon^2$.
\end{theorem}

Before proving Theorem \ref{teo6.2}, we need to recall a couple of classical results.  To begin with, we fix the following notation: by $C^\infty_0([a,b])$ we mean the members of $C^\infty_0(\R)$ whose support is contained in the set $[a,b]$, with $a<b$.

\begin{lemma}\label{lema6.1}
Let $\al$ be a continuous function on $[a,b]$. Assume that 
$$\int_{a}^{b}\al(x)f(x)=0,$$
for every continuous function $f$. Then $\al(x)=0$, for all $x\in[a,b]$.
\end{lemma}
\begin{proof}
See \cite{fomin}, Lemma 1, on page 9.
\end{proof}

\begin{lemma}\label{lema6.3}
Let $A,\epsilon$ be two real constants such that $A\epsilon\neq0$, and $\phi:\R\rightarrow\R$ be defined by $\phi(z)=A\, e^{-|z|/\epsilon}$.
Then $\phi'(z)=-\sign{z}\,\phi(z)/\epsilon$ in the weak sense. In particular, $(\phi'(z))^2=\phi(z)^2/\epsilon^2$.
\end{lemma}

\begin{proof}
Let $\phi$ be any test function. Then
$$
\ba{lcl}
\ds{\int_{-\infty}^{+\infty}{\phi(z)\psi'(x)}dx}&=&\ds{\int_{-\infty}^{0}{Ae^{z/\epsilon}\psi'(z)}dz+\int^{+\infty}_{0}{Ae^{-z/\epsilon}\psi'(z)}dz}\\
\\
&=&\ds{-\frac{1}{\epsilon}\int_{-\infty}^{0}{Ae^{z/\epsilon}\psi(z)}dz+\f{1}{\epsilon}\int^{+\infty}_{0}{Ae^{-z/\epsilon}\psi(z)}dz}\\
\\
&=&\ds{\int_{-\infty}^{+\infty}\f{\sign(z)}{\epsilon}\phi(z)\psi(z)dz}
\ea
$$
and $\phi(z)=-\sign{(z)}\phi(z)/\epsilon$.
\end{proof}

\begin{lemma}\label{lema6.2}
Under the conditions of Lemma \ref{lema6.3}, if
$\psi\in{\cal S}(\R)$ and
$$J_n(\psi):=\int_{-\infty}^\infty \phi^n(z)\psi(z)dz\quad\text{and}\quad I_n(\psi):=\int_{-\infty}^\infty \phi^n(z)\psi''(z)dz,$$
then
$$I_n(\psi)=-2\left(\f{n}{\epsilon}\right)A^n\psi(0)+\left(\f{n}{\epsilon}\right)^2J_n(\psi).$$
\end{lemma}

\begin{proof}
We only need to integrate $I_n(\psi)$ by parts to obtain the result.
\end{proof}

We recall that if $\phi$ is a distribution, its $n$th derivative $\phi^{(n)}$ satisfies the relation
$$\int_{-\infty}^{+\infty} \phi(z)\psi^{(n)}(z)dz=(-1)^n\int_{-\infty}^{+\infty} \phi^{(n)}(z)\psi(z)dz,$$
for any test function $\psi$.



We are bound to prove Theorem \ref{teo6.2}. To do it, we assume that $u(x,t)=Ae^{-|z|/\epsilon}$, with $z=x-ct$, is a solution of (\ref{1.0.4}). 
Among the many alternatives to prove Theorem \ref{teo6.2}, we will consider once more our 
second conservation law given in Theorem \ref{teo4.2}. More precisely, we will examine equation \eqref{5.1.2} after plugging the function $\phi(z)=Ae^{-|z|/\epsilon}$ in such a way that the peakon solution will be forced to satisfy the resulting equation with the right hand side equals $0$.

Taking $A=0$ in equation (\ref{5.1.2}), from Lemma \ref{lema6.3} we have the following weak formulation for (\ref{5.1.2}):
\bb\label{6.2.2}
\int_{-\infty}^{+\infty}\left((\al+c)\phi-2\phi^2+\f{\beta}{3}\phi^3+\f{\gamma}{4}\phi^4\right)\psi dz+\f{\epsilon^2}{2}\int_{-\infty}^{+\infty}\phi^2\psi''dz+(\Gamma-c\epsilon^2)\int_{-\infty}^{+\infty}\phi\psi''dz=0,
\ee
for any test function $\psi\in{\cal S}(\R)$. The demonstration is concluded if we prove that the function $\phi(z)=A e^{-|z|/\epsilon}$ solves Eq. (\ref{6.2.2}) if and only if $\be=\gamma=0$, $\Gamma=-\al\epsilon^2$ and $A=\al+c$.

\begin{proof}({\bf of Theorem \ref{teo6.2}}): 
By Lemma \ref{lema6.2}, we can rewrite equation (\ref{6.2.2}) as
\bb\label{6.2.3}
2A\left(\epsilon A+\f{\Gamma-c\epsilon^2}{\epsilon}\right)\psi(0)-\f{\Gamma+\al\epsilon^2}{\epsilon^2}J_1-\f{\be}{3}J_3-\f{\gamma}{4}J_4=0.
\ee
Observe that the contribution of $J_2$ vanishes, once the coefficients of this term cancel one to each other. Finally, equation (\ref{6.2.3}) can be rewritten as
\bb\label{6.2.4}
2A\left(\epsilon A+\f{\Gamma-c\epsilon^2}{\epsilon}\right)\psi(0)-\int_{-\infty}^{+\infty}\left(\f{\Gamma+\al\epsilon^2}{\epsilon^2}\phi+\f{\be}{3}\phi^3+\f{\gamma}{4}\phi^4\right)\psi dz=0.
\ee
Note now that equation (\ref{6.2.4}) must be valid for any test function $\psi$. Restricting ourselves firstly to $\psi\in C^{\infty}_0([a,b])$, with $0<a<b$ or $a<b<0$, we are forced to conclude, in view of Lemma \ref{lema6.1}, that 
$$
\f{\Gamma+\al\epsilon^2}{\epsilon^2}\phi+\f{\be}{3}\phi^3+\f{\gamma}{4}\phi^5=0,
$$
which implies that $\Gamma=-\al\epsilon^2$, $\be=0$ and $\gamma=0$. Then, with these restrictions, equation (\ref{6.2.3}) is consistent provided that $A=\al+c$. This proves the {\it if} part. The {\it only if} part is proved noticing that if $\be=\gamma=0$, $\Gamma=-\al\epsilon^2$ and $A=\al+c$, then equation (\ref{6.2.4}) is automatically satisfied, as well as equation (\ref{6.2.2}).
\end{proof}

\begin{remark}\label{rem6.2}
Under the restrictions of Theorem \ref{teo6.2}, equation $(\ref{1.0.4})$ can be transformed into the Camassa-Holm equation under the change $u\mapsto u-\al$. Therefore, not only the equation has peakon solutions as the one given in Theorem \ref{teo6.2}, but also multipeakon solutions \cite{chprl}. We, however, do not consider them here because they are, in essence, the same of the Camassa-Holm equation, taking into account weak derivatives and the shift $u\mapsto u-\al$.
\end{remark}

\begin{remark}
Theorem $\ref{teo6.2}$ assures that equation $(\ref{1.0.4})$ has peakon solutions of the type $u(x,t)=A e^{\lambda|x-ct|}$ {\it if and only if} $A=\al+c$ and $\lambda=-1/\epsilon$. However, the same theorem {\it does not imply} on the nonexistence of other peakon solutions. In fact, Theorem \ref{teo5.8} assures the existence of other peakon solutions, such as the periodic ones. Combining Theorem \ref{teo5.8} and Theorem $\ref{teo6.2}$, we conclude that if $u=u(x,t)$ is a peakon solution of $(\ref{1.0.4})$ with $\be\neq0$, then $u$ cannot be of the form $u(x,t)=(\al+c)e^{-|x-ct|/\epsilon}$. 
\end{remark}

\section{Members describing pseudo-spherical surfaces}\label{sec7}
Here we investigate the geometric integrability of equation \eqref{1.0.11}. Following Reyes \cite{reyes2002}, an equation is geometrically integrable if it describes a non-trivial one-parameter family of pseudo-spherical surfaces.

A two-dimensional manifold ${\cal M}$ is a PSS if there exist one-forms $\omega_1,\,\omega_2,\,\omega_3$ on ${\cal M}$ such that $\omega_1\wedge\omega_2\neq0$ and
\bb\label{7.0.1}
d\omega_1=\omega_3\wedge\omega_2,\quad d\omega_2=\omega_1\wedge\omega_3,\quad d\omega_3=\omega_1\wedge\omega_2.
\ee
Conditions \eqref{7.0.1} are structure equations of ${\cal M}$, which endows ${\cal M}$ with the metric $ds^2=\omega_1^2+\omega_2^2$ having constant Gaussian curvature ${\cal K}=-1$.

Assume that $F[u_{(n)}]=0$ is a differential equation, where the unknown $u$ is assumed to depend on two variables, say $x$ and $t$. This equation is said to describe a PSS if there are one-forms \eqref{1.0.11}, with smooth coefficient functions, such that the triple $\{\omega_1,\omega_2,\omega_3\}$ satisfies \eqref{7.0.1} whenever $u$ is a solution of $F=0$. We note that conditions \eqref{7.0.1} can be rewritten as
$$d\Omega=\Omega\wedge\Omega,$$
where $\Omega$ is the matrix whose entries are one-forms, namely,
\bb\label{7.0.2}\Omega=\f{1}{2}\begin{pmatrix}
\omega_2 & \omega_1-\omega_3 \\
\omega_1-\omega_3 & -\omega_2
\end{pmatrix}=:(\Omega_{ij}),
\ee
and $\Omega\wedge\Omega:=(\sum_{k}\Omega_{ik}\wedge\Omega_{kj})$ and $d\Omega:=(d\Omega_{ij})$. In view of \eqref{7.0.2}, we can see that $\omega_2$ is related to the spectral parameter from the theory of integrable systems, see \cite{chern}. For further details about equations describing PSS, see \cite{diego1,diego2,chern,keti1992,reyes2002,reyes2011,keti2015}. The interested reader is also referred to \cite{rogers1982,rogers2002}, where this topic is extensively covered.

\subsection{Technical results}

We observe, from a pragmatic point of view, that the coefficient functions of the one-forms are arbitrary, which may bring (and it really does!) practical and computational complications. This is usually overcome imposing some restrictions on the differential functions involved \cite{keti1992,keti2015}. Therefore, in our analyses we employ the same restrictions used in \cite{diego1,diego2,keti1992,keti2015} to investigate the geometric integrability of third order equations describing PSS. 

We have the following result from \cite{keti1992}.
\begin{lemma}\label{lema7.1}
Let $u_t=u_{xxx}+G(u,u_x,u_{xx})$ be a differential equation describing $PSS$ with associate one-forms $\eqref{1.0.11}$ satisfying $\omega_2=\lambda dx+f_{22}dt$, in which $\lambda$ is a real parameter. Then $G$ is independent of $\lambda$ if and only if, up to a change of the dependent variable, the aforementioned equation is one of the following:
\bb\label{7.1.1}
\ba{l}
u_t=u_{xxx}+au^2u_x+buu_x+cu_x,\\
\\
u_t=u_{xxx}+a u_{xx}-3b(uu_x){x}+buu_{x}(3bu-2a),\\
\\
u_t=(u_x-\ell(u))_{xx}-[(\delta u+\mu)u_x-\ell(u)]_x,\\
\\
u_t=u_{xxx}+au_{xx}+bu_{x}+cu+\eta,
\ea
\ee
where $a,\,b,\,c,\,\delta,\,\mu$ and $\eta$ are constants, with $ab\delta\neq0$, and $\ell(u)$ is a differentiable function.
\end{lemma}
\begin{proof}
See \cite{keti1992}, theorems 3.1 and 4.1. See also Theorem 2.13 of \cite{reyes2011}.
\end{proof}

Our next result is an extremely useful one, but the notation might be very hard and difficult. Therefore, in order to make it easier as possible, we will follow conventions used in \cite{chern,keti1992,keti2015}: $z_0:=u$, $z_1:=u_x$, $z_2:=u_{xx}$ and $z_3=u_{xxx}$.
\begin{lemma}\label{lema7.2}
A necessary condition for an equation
\bb\label{7.1.2}
z_{0,t}-z_{2,t}=\lambda\,z_0z_3+G(z_0,z_1,z_2),\quad G\neq0,\,\lambda\neq0
\ee
to describe $PSS$ with associated one-forms \eqref{1.0.11} satisfying 
$$f_{i1}=\mu_i\,f_{11}+\eta_i,\quad \mu_i,\,\eta_i\in\R,\,\,i=2,3,$$
is the existence of smooth functions $\psi=\psi(z_0,z_1)$ and $h=h(z_0-z_2)$, with $h'\neq0$, and these two functions and $G$ satisfy at least one of the following conditions:
\begin{enumerate}
    \item $G=(z_1\,\psi_{z_0}+z_2\psi_{z_1}+m\,\psi)/h'$, where $0\neq m\in\R$;
    \item $G=-\lambda (z_1\,h+z_0z_1h'+m_1z_1+m_2z_2)/h'$, where $m_1,m_2\in\R$ and $m_2\neq0$;
    \item $G=[z_2\,\psi_{z_1}+z_1\psi_{z_0}+m_1\,\psi-\lambda z_0 z_1 h'-(\lambda z_1+\lambda m_1 z_0+m_2)h]/h'$, where $m_1$ and $m_2$ are constants such that $(m_1,m_2)\neq(0,0)$;
    \item $G=\lambda(z_1z_2-2z_0z_1-m z_1/\tau\mp z_2/\tau)+\tau e^{\pm \tau z_1}(\tau z_0 z_2\pm z_1+m z_2)\varphi(z_0)\pm e^{\pm\tau z_1}(\tau z_0 z_1+\tau z_1 z_2+m z_1\pm z_2)\varphi'(z_0)+z_1^2 e^{\pm \tau z_1}\varphi''(z_0)$, where $m,\tau\in\R$, $\varphi(z_0)\neq0$ and $\tau>0$;
    \item $G=\lambda(2z_1z_2-3z_0z_1-m_2z_1)+m_1\theta e^{\theta z_0}(\theta z_1^3+z_1z_2+2z_0z_1+m_2z_1)$, where $m_1,\,m_2,\,\theta\in\R$, with $\theta\neq0$.
\end{enumerate}
\end{lemma}
\begin{proof}
These are conditions requested in theorems 3.2, 3.3, 3.4 and 3.5 of \cite{keti2015} for the equation describes pseudo-spherical surfaces. Therefore, we omit the proof and refer the reader to the original work by Silva and Tenenblat.
\end{proof}

\begin{lemma}\label{lema7.3}
Assume that the equation $u_t-u_{txx}=\lambda u u_{xxx}+G(u,u_x,u_{xx})$
is such that condition $5$ in Lemma $\ref{lema7.2}$ is satisfied. Then the function $f_{ij}$ in $\eqref{1.0.11}$ are given by
$$f_{11}=a(z_0-z_2)+b,\quad f_{21}=\mu f_{11}+\eta,\quad f_{31}=\pm\sqrt{1+\mu^2}f_{11}\pm\f{\theta\pm a\,\eta\,\mu}{1+\mu^2},$$
$$f_{12}=-\lambda\,z_0\,f_{11}+a\,m_1\theta\,e^{\theta\,z_0}z_1^2+(m_1\,\theta\,e^{\theta\,z_0}-\lambda)\left[\f{a\,z_0+b}{\theta}\pm\left(\mu-\f{a\,\eta}{\theta}\right)\f{z_1}{\sqrt{1+\mu^2}}\right],$$
$$f_{22}=-\lambda\,z_0\,f_{21}+\mu\,a\,m_1\,\theta\,e^{\theta\,z_0}z_1^2+\f{m_1\theta\,e^{\theta\,z_0}-\lambda}{\theta}\left[\mu(az_0+b)+\eta\mp(\theta-\mu\,a\,\eta)\f{z_1}{\sqrt{1+\mu^2}}\right],$$
$$
\ba{lcl}
f_{32}&=&\ds{-\lambda\,z_0\,f_{31}\pm\sqrt{1+\mu^2}\,a\,m_1\,\theta\,e^{\theta\,z_0}z_1^2}\\
\\
&&\ds{-\left(\f{m_1\,\theta\,e^{\theta\,z_0}-\lambda}{\theta}\right)\left\{a\,\eta\,z_1\mp\f{1}{\sqrt{1+\mu^2}}\left[(1+\mu^2)(az_0+b)+\mu\eta+\f{\theta}{a}\right]\right\},}
\ea
$$
and the constants $a$, $b$, $\mu$, $\eta$, $\theta$ and $m_2$ are related by the relation
\bb\label{7.1.3}
b=\f{a}{2\theta}\left[\f{(\mu\theta-a\eta)^2}{a^2(1+\mu^2)}-\f{a}{\theta}+m_2\theta-1\right].
\ee
\end{lemma}

We observe that the parameter $\lambda$ in Lemma \ref{lema7.1} and $\eta_2$, in Lemma \ref{lema7.2}, are related to the spectral parameter in the literature of integrable systems \cite{ablo1,ablo2}.

\subsection{Proof of Theorem \ref{teo1.3}}\label{subsec7.2}

Let us now assume $\epsilon\neq 0$. Making the changes $x\mapsto x/\epsilon, t\mapsto t/\epsilon$ and next $u\mapsto u-\Gamma/\epsilon^2$, \eqref{1.0.4} is transformed into, after renaming the constants,
\bb\label{7.2.1}
u_t-u_{txx}=uu_{xxx}+2u_xu_{xx}+(\al-3u+\beta u^2+\gamma u^3)u_x,
\ee
which is of the type \eqref{7.1.2} with $G(z_0,z_1,z_2)=2z_1z_2+(\al-3z_0+\be z_0^2+\gamma z_0^3)z_1$ and $\lambda=1$.

We have the following preliminary result:
\begin{proposition}\label{prop7.1}
The only condition in Lemma $\ref{lema7.2}$ satisfied by equation $\eqref{7.2.1}$ is number $5$ if and only if $\be=\gamma=0$ and $m_2=-\al$.
\end{proposition}

\begin{proof}
Substituting $G(z_0,z_1,z_2)=2z_1z_2+(\al-3z_0+\be z_0^2+\gamma z_0^3)z_1$ and $\lambda=1$ into the possible forms listed in Lemma \ref{lema7.2}, we will conclude that conditions $1$ to $4$ lead to contradictions. Substituting into the last one, we have the following identity:
$$
2z_1z_2-3z_0z_1-m_2z_1+m_1\theta e^{\theta z_0}(\theta z_1^3+z_1z_2+2z_0z_1+m_2z_1)=2z_1z_2+(\al-3z_0+\be z_0^2+\gamma z_0^3)z_1,
$$
which implies
$$
-m_2z_1+m_1\theta e^{\theta z_0}(\theta z_1^3+z_1z_2+2z_0z_1+m_2z_1)=\al z_1+\be z_0^2z_1+\gamma z_0^3z_1.$$

 From the coefficient of $z_1^3$ we conclude that $m_1\theta=0$, which implies that $m_1=0$ in view of the constraints in Lemma \ref{lema7.2}. As a consequence, $m_2=-\al$ and $\be=\gamma=0$.
\end{proof}

\begin{corollary}\label{cora}
Under the restrictions from Proposition $\ref{prop7.1}$, the constraint $(\ref{7.1.3})$ can be reduced to
$b=-1+(\eta^2-\alpha)/2$.
\end{corollary}

\begin{proof}
By Proposition $\ref{prop7.1}$, we have $m_1=0$ and $m_2=-\alpha$. Choosing $a=1$, $\theta=1$ and $\mu=0$ we obtain the desired result.
\end{proof}

{\bf Proof of Theorem \ref{teo1.3}}: Let us begin with the case $\epsilon=0$. Making a suitable choice of the constant $u_0$, the shift $u\mapsto u-u_0$ transforms the resulting equation (after renaming constants) in the first equation in \eqref{7.1.1}. The result of this part is then an immediate consequence of Lemma \ref{lema7.1}.

Now assume $\epsilon\neq0$. As previously shown, we can transform equation \eqref{1.0.4}, with $\epsilon\neq0$, in \eqref{7.2.1}. By Proposition \ref{prop7.1}, we conclude that \eqref{1.0.4} describes a PSS if and only if $\be=\gamma=0$. Substituting $a=\theta=1$, $\mu=m_1=0$ and $m_2=-\alpha$ into Lemma \ref{lema7.3} we conclude the proof of Theorem \ref{teo1.3}.




\section{Discussion}\label{sec8}

Our work about equation \eqref{1.0.4} was motivated by the recent equation \eqref{1.0.1} introduced in \cite{chines-jde} and later considered in \cite{chines-arxiv}. These two papers led us to consider equation \eqref{1.0.4}, which is structurally the same as \eqref{1.0.1}, but without the physical constraints given by \eqref{1.0.2}.

In our case, as mentioned in the Introduction, equation \eqref{1.0.4} not only is a mathematical generalization of the physical model \eqref{1.0.1}, but also splits in two significantly large families of equations (after suitable scalings, shifts and eventually renaming the constants) 
\bb\label{8.0.1}
u_t=u_{xxx}+(\al-3u+\be u^2+\gamma u^3)u_x
\ee
and
\bb\label{8.0.2}
u_t-u_{txx}=uu_{xxx}+2u_xu_{xx}+(\al-3u+\be u^2+\gamma u^3),
\ee
where we took $\epsilon=1$ for convenience. 
This division is, in particular, evident in the demonstration of local well-posedness (Theorem \ref{teo1.1}), conservation laws and its consequences (sections \ref{sec2}, \ref{sec3} and \ref{sec4}), and of extreme importance in the classification of travelling wave solutions carried out in Section \ref{sec5}. Moreover, if $\epsilon=0$ we do not have the emergence of ``purely'' weak solutions, 
whereas the case $\epsilon\neq0$ has peakons and cuspons (both periodic and non-periodic), as shown in theorems \ref{teo5.5} and \ref{teo5.6}. In addition, the classification of members of \eqref{1.0.4} describing PSS also depends on the values of $\epsilon$ that, although not the only relevant parameter, is certainly the most special one.

Sometime after we initiated our investigation, we discovered the reference \cite{chines-adv}, where the local well-posedness of an equation analogous to \eqref{1.0.1} is claimed, but omitted. We then proved the local well-posedness for \eqref{1.0.4}. To pursue this goal, we removed some restrictions in the main result of \cite{liu2011} and as a consequence of Theorem \ref{teo1.2} we not only recover the local well-posedness to the CH equation, {\it e.g}, see \cite{escher, mustafa,blanco}, but also obtain corollaries \ref{cor2.1} and \ref{cor2.2}, which are nothing but the results proved in \cite{liu2011} and \cite{mustafa}. Moreover, as a consequence of Theorem \ref{teo1.2}, if $u_0\in H^{s}(\R)$, $s>3/2$, and $h$ is a $C^\infty(\R)$ function such that $h(0)=0$, then the initial value problem
\bb\label{8.0.3}
\left\{\ba{l}
\ds{u_t-u_{txx}+\p_x h(u)=\p_x\left(\f{e^u}{2}u_x^2+e^uu_{xx}\right)},\\
\\
u(x,0)=u_0(x),
\ea\right.
\ee
is well-posed in $H^s(\R)$. Although this fact is a trivial consequence of our theorem, it cannot be recovered by invoking  corollaries \ref{cor2.1} and \ref{cor2.2} of \cite{liu2011} and \cite{mustafa}, respectively. This example shows that Theorem \ref{teo1.2} is not a small improvement on the mentioned results, but actually an extension that, in particular, implies the local well-posedness of the problem \eqref{1.0.5}. We observe that we were able to remove the condition $g(0)=0$ in the results proved in \cite{liu2011}, but we could not do the same regarding the function $h$ in corollaries \ref{cor2.1} and \ref{cor2.2}. We conjecture that it might be possible by following similar steps of Theorem \ref{teo1.2} to relax the conditions for $g$, but this is an open problem to be considered in another moment.

We established conservation laws for equation \eqref{1.0.4} with $\epsilon\neq 0$ in Section \ref{sec3}. The case $\epsilon=0$ was not considered because the choice transforms the equation into a three-parameter evolution equation, and the literature about this sort of equation is very vast. It is well-known that the KdV and other evolution equations have an infinite hierarchy of symmetries and conservation laws \cite{ablo1,ablo2,igor2012,igor2014,miura,olverbook,popPLA,rita,arturrmp}, meaning that a study of these sort of equations would not bring anything new, see, for instance, \cite{popPLA}. Moreover, some of the equations in \eqref{1.0.4} with $\epsilon=0$ can be derived as hydrodynamical models \cite{ablo1,ablo2,burde}, which usually have some conservation laws as a consequence of their physical backgrounds.

On the other hand, in the case $\epsilon\neq0$, the conservation laws enabled us to obtain several properties of the solutions, as presented in Section \ref{sec4}, but more importantly, led us to obtain the quadrature \eqref{5.1.4}, which makes a classification of bounded travelling wave solutions of \eqref{1.0.4} possible, as shown in Section \ref{sec5}.

In \cite{gui-jnl} the authors performed a classification of bounded travelling wave solutions of equation \eqref{1.0.1}. However, their classification is a very particular case of ours because in \cite{gui-jnl} the authors reduced the quadrature \eqref{5.1.4} to the particular case of two real zeros. In fact, they considered a quadrature with the polynomial in \eqref{5.1.4} replaced by the one given in \eqref{1.0.9}. As a consequence of this simplification, the periodic travelling wave solutions do not appear in their classification. Our results, therefore, not only recover the one proved in \cite{gui-jnl}, but also describe periodic travelling wave solutions admitted by \eqref{1.0.4}. 
Opposed to \cite{len,raspa}, we have not found the composite travelling wave solutions of \eqref{1.0.4} because the computations required needed a deeper understanding of certain underlying geometric conditions that are beyond the scope of this paper. 
We observe that we have classified 139 possible cases of travelling wave solutions of \eqref{1.0.4}, whereas in \cite{gui-jnl} only 15 were obtained. Our classification includes the case $\epsilon=0$, while the one carried out in \cite{gui-jnl} only considers $\epsilon\neq0$. Even if we restrict our results for $\epsilon\neq 0$ we still have a larger classification when compared with the one presented in \cite{gui-jnl}, since we also analyse the existence of periodic solutions.

We also investigated the existence of members in \eqref{1.0.4} describing PSS. For $\epsilon=0$ this is a mere consequence of the results proved by Rabelo and Tenenblat \cite{keti1992}, while for the case $\epsilon\neq0$ we have a more interesting and rich situation. In case this condition is satisfied, equation \eqref{1.0.4} can be transformed into equation \eqref{8.0.2}. This equation includes the CH equation and its relations with PSSs were firstly investigated in \cite{reyes2002}, see also \cite{reyes2011}. The question that remained to be investigated was if there were other members of \eqref{8.0.2} (or \eqref{1.0.4}) that would also have this property. The answer is given by Theorem \ref{teo1.3}.

A natural and interesting question that emerges from Theorem \ref{teo1.3} is if equation \eqref{1.0.1} might describe PSS. The answer is positive, but with limitations on the values given by \eqref{1.0.2}. Making the shift $u\mapsto u-\be_0/\be$ in \eqref{1.0.1}, we have the following equation
\bb\label{8.0.4}
\ba{lcl}
\ds{u_t-u_{txx}+3uu_x-2u_xu_{xx}-uu_{xxx}}&=&\ds{\left(\f{\omega_2\be^3_0}{\be^3\al^3}-c+3\f{\be_0}{\be}-\f{\omega_1\be_0^2}{\al^2\be}\right)u_x}\\
\\
&&\ds{+\f{\be_0}{\be}\left(2\f{\omega_1}{\al^2}-3\f{\be_0}{\be}\right)uu_x+\left(3\f{\be_0\omega_2}{\be\al^3}-\f{\omega_1}{\al^2}\right)u^2u_x}\\
\\
&&\ds{-\f{\omega_2}{\al^3}u^3u_x}.
\ea
\ee
In view of the constraints \eqref{1.0.2}, if $c^2=1$ or $c^2=2$, then $\omega_1=\omega_2=0$ and equation \eqref{8.0.4} satisfies the conditions of Theorem \ref{teo1.3}. If $c^2=1$, then $\Omega=0$ from \eqref{1.0.2}, which is equivalent to say that we do not have Coriollis effect in \eqref{1.0.1}. If $c^2=2$, then $\Omega=\pm\sqrt{2}/4$. The negative value must be discarded since $\Omega$ is a positive physical variable \cite{gui-jnl}.

\section{Conclusion}\label{sec9}

In this paper we generalized a previous result in \cite{liu2011} regarding the local well-posedness of equations of the type \eqref{1.0.7}. As a consequence of this generalization, we have immediately assured local well-posedness of the Cauchy problem \eqref{1.0.5}. We also found some conservation laws for \eqref{1.0.4}, which provided us some qualitative information about the solution of the equation, see theorems \ref{teo4.1}, \ref{teo4.2}, corollaries \ref{cor4.1}, \ref{cor4.2} and, more importantly, theorems \ref{teo5.1}--\ref{teo5.10}. Moreover, we also determined the constraints on the parameters in \eqref{1.0.4} that would enable the emergence of peaked solutions (see Theorem \ref{teo6.2}) as weak soliton solutions. Finally, we also classified the members of \eqref{1.0.4} that can describe pseudo-spherical surfaces.

\section*{Acknowledgements}

The idea of this paper occurred during the period the second author was as a visiting professor at Silesian University in Opava, and the majority of the theorems reported in this paper were proven there. Both authors gratefully acknowledge the warm hospitality of Silesian University in Opava, Czech Republic, during the time of their respective visits.

P. L. da Silva would like to thank \textit{Programa de Pós-Graduação em Matemática} from \textit{Departamento de Matemática of Universidade Federal de São Carlos} for the financial support to her short research visit in the Silesian University in Opava. P. L. da Silva is also thankful to CAPES for her post-doctoral fellowship.

I. L. Freire is thankful to the participants of the \textit{Seminář z diferenciální geometrie a jejich aplikací} of the Department of Geometry and Mathematical Physics of the Silesian University in Opava, where parts of the results of this paper were presented. Special thanks are given to Prof. A. Sergyeyev, Prof. R. Popovych, Prof. D. Catalano Ferraioli and Prof. M. F. da Silva for their useful discussions and comments. I. L. Freire's work is partially supported by CNPq (grants 308516/2016-8 and 404912/2016-8).

Finally, both authors would like to express their deep gratitude to the reviewer for their acute and detailed reading of the manuscript, which resulted in a considerable improvement of the paper.

\end{document}